\documentclass[11pt]{article}
\usepackage{tikz} 
\usetikzlibrary{decorations,arrows} 
\usetikzlibrary{arrows,decorations.markings}
\usepackage[margin=1in]{geometry}
\usepackage{color}
\usepackage{amssymb}
\usepackage{amsmath}
\usepackage{amsthm}
\usepackage{url}
\usepackage{enumitem}
\usepackage{amsfonts}
\usepackage{hyperref}
\usepackage{graphicx}
\usepackage[norelsize,ruled,vlined]{algorithm2e}
\usepackage{authblk}

\newtheorem{theorem}{Theorem}
\newtheorem{definition}[theorem]{Definition}
\newtheorem{proposition}[theorem]{Proposition}
\newtheorem{lemma}[theorem]{Lemma}
\newtheorem{corollary}[theorem]{Corollary}

\newtheorem{problem}[theorem]{Problem}

\newcommand\T{\rule{0pt}{2.6ex}}
\newcommand\B{\rule[-1.2ex]{0pt}{0pt}}

\begin{document}

\title{Probably certifiably correct $k$-means clustering}
\author[1]{Takayuki Iguchi}
\author[1]{Dustin G.\ Mixon}
\author[1]{Jesse Peterson}
\author[2]{Soledad Villar}
\affil[1]{Department of Mathematics and Statistics, Air Force Institute of Technology}
\affil[2]{Department of Mathematics, University of Texas at Austin}

\date{}
\maketitle

\begin{abstract}
Recently, Bandeira~\cite{bandeira2015note} introduced a new type of algorithm (the so-called probably certifiably correct algorithm) that combines fast solvers with the optimality certificates provided by convex relaxations.
In this paper, we devise such an algorithm for the problem of $k$-means clustering.
First, we prove that Peng and Wei's semidefinite relaxation of $k$-means~\cite{peng2007approximating} is tight with high probability under a distribution of planted clusters called the stochastic ball model.
Our proof follows from a new dual certificate for integral solutions of this semidefinite program.
Next, we show how to test the optimality of a proposed $k$-means solution using this dual certificate in quasilinear time.
Finally, we analyze a version of spectral clustering from Peng and Wei~\cite{peng2007approximating} that is designed to solve $k$-means in the case of two clusters.
In particular, we show that this quasilinear-time method typically recovers planted clusters under the stochastic ball model.
\end{abstract}

\section{Introduction}

Clustering is a central problem in unsupervised machine learning.
It consists of partitioning a given finite sequence $\{x_i\}_{i=1}^N$ of points in $\mathbb{R}^m$ into $k$ subsequences such that some dissimilarity function is minimized.
Usually, this function is chosen ad hoc with an application in mind.
A particularly common choice is the \textbf{$k$-means objective}:
\begin{alignat}{2}
\label{kmeans}
&\text{minimize}  &       &
\sum_{t=1}^k \sum_{i\in A_t} \bigg\| x_i-\frac{1}{|A_t|}\sum_{j\in A_t} x_j \bigg\|_2^2\\
\nonumber
& \text{subject to}& \quad & A_1\sqcup\cdots\sqcup A_k=\{1,\ldots,N\}
\end{alignat}
Problem \eqref{kmeans} is NP-hard in general~\cite{jms06}.
A popular heuristic for solving $k$-means is Lloyd's algorithm, also known as the $k$-means algorithm~\cite{Lloyd}.
This algorithm alternates between calculating centroids of proto-clusters and reassigning points according to the nearest centroid.
In general, Lloyd's algorithm (and its variants~\cite{Arthur07, Lloyd06}) may converge to local minima of the $k$-means objective (e.g., see section~5 of \cite{relax}).
Furthermore, the output of Lloyd's algorithm does not indicate how far it is from optimal.
Instead, we seek a new sort of algorithm, recently introduced by Bandeira~\cite{bandeira2015note}:

\begin{definition}
Let $\mathbf{P}$ be an optimization problem that depends on some input, and let $\mathbf{D}$ denote a probability distribution over possible inputs.
Then a \textbf{probably certifiably correct (PCC) algorithm} for $(\mathbf{P},\mathbf{D})$ is an algorithm that on input $D\sim \mathbf{D}$ produces a global optimizer of $\mathbf{P}$ with high probability, and furthermore produces a certificate of having done so.
\end{definition}

Most non-convex optimization methods fail to produce a certificate of global optimality.
However, if a non-convex problem enjoys a convex relaxation, then solving the dual of this relaxation will produce a certificate of (approximate) optimality.
Along these lines, the $k$-means problem enjoys a semidefinite relaxation.
To see this, let $1_A$ denote the indicator function of $A\subseteq\{1,\ldots,N\}$, and define the $N\times N$ matrix $D$ by $D_{ij}:=\|x_i-x_j\|_2^2$.
Then taking
\begin{equation}
\label{eq.integral}
X:=\sum_{t=1}^k\frac{1}{|A_t|}1_{A_t}1_{A_t}^\top,
\end{equation}
the $k$-means objective \eqref{kmeans} may be expressed as $\frac{1}{2}\operatorname{Tr}(DX)$.
Since $X$ satisfies several convex constraints, we may relax the region of optimization to produce a convex program, namely, the Peng--Wei semidefinite relaxation~\cite{peng2007approximating} (cf.\ \cite{chen20080}):
\begin{alignat}{2}
\label{eq.kmeansSDP}
& \text{minimize}  &       & \operatorname{Tr}(DX) \\
\nonumber
& \text{subject to}& \quad & 
\begin{aligned}[t]
\operatorname{Tr}(X)&=k,~
X1=1,~
X\geq0,~
X\succeq0
\end{aligned}
\end{alignat}
Here, $X\geq0$ means that each entry of $X$ is nonnegative, whereas $X\succeq0$ means that $X$ is symmetric and positive semidefinite.

Recently, it was shown that under a certain random data model, this convex relaxation is \textbf{tight} with high probability~\cite{relax}, that is, every solution to the relaxed problem \eqref{eq.kmeansSDP} has the form \eqref{eq.integral}, thereby identifying an optimal clustering.
As such, in this high-probability event, one may solve the dual program to produce a certificate of optimality.
However, semidefinite programming (SDP) solvers are notoriously slow.
For example, running MATLAB's built-in implementation of Lloyd's algorithm on 64 points in $\mathbb{R}^6$ will take about 0.001 seconds, whereas a CVX implementation~\cite{grant2008cvx} of the dual of \eqref{eq.kmeansSDP} for the same data takes about 20 seconds.
Also, Lloyd's algorithm scales much better than SDP solvers, and so one should expect this runtime disparity to only increase with larger datasets.
Overall, while the SDP relaxation theoretically produces a certificate in polynomial time (e.g., by an interior-point method \cite{nesterov1994interior}), it is far too slow to wait for in practice.

As a fast alternative, Bandeira~\cite{bandeira2015note} recently devised a general technique to certify global optimality.
This technique leverages several components simultaneously:
\begin{itemize}
\item[(i)]
A fast non-convex solver that produces the optimal solution with high probability (under some reasonable probability distribution of problem instances).
\item[(ii)]
A convex relaxation that is tight with high probability (under the same distribution).
\item[(iii)]
A fast method of computing a certificate of global optimality for the output of the non-convex solver in (i) by exploiting convex duality with the relaxation in (ii).
\end{itemize}
In the context of $k$-means, one might expect Lloyd's algorithm and the Peng--Wei SDP to be suitable choices for (i) and (ii), respectively.
For (iii), one might adapt Bandeira's original method in~\cite{bandeira2015note} based on complementary slackness (see Figure~\ref{figure.pcc} for an illustration).
In this paper, we provide a theoretical basis for each of these components in the context of $k$-means.

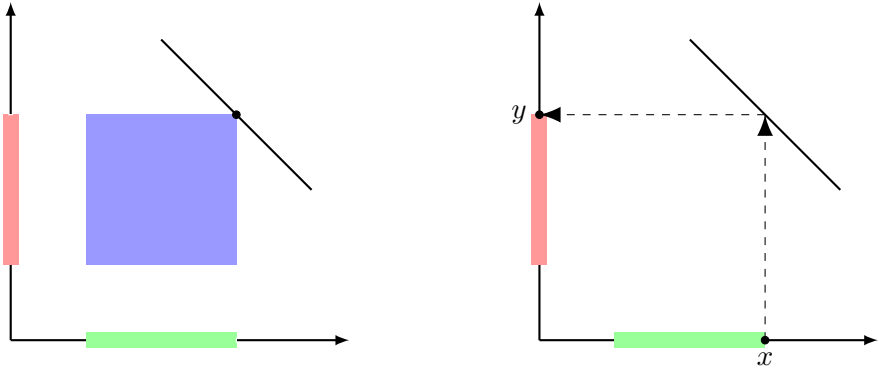
\begin{figure}
\begin{center}
\begin{tikzpicture}
\draw[thick,->,>=latex ] (0,0) -- (4.5,0);
\draw[thick,->, >=latex] (0,0) -- (0,4.5);
\fill[fill=blue!40!white] (1,1) rectangle (3,3);
\fill[fill=green!40!white] (1,-0.1) rectangle (3,0.1);
\fill[fill=red!40!white] (-0.1,1) rectangle (0.1,3);
\draw[thick] (2,4) -- (4,2);
\draw[white] (1pt,3 cm) -- (-1pt,3 cm) node[anchor=east] {$y$};
\filldraw[fill=black, draw=black] (3,3) circle (.5mm);
\draw[white] (3 cm,1pt) -- (3 cm,-1pt) node[anchor=north] {$x$};
\end{tikzpicture}
\hspace{50pt}
\begin{tikzpicture}
\draw[thick,->, >=latex] (0,0) -- (4.5,0);
\draw[thick,->, >=latex] (0,0) -- (0,4.5);
\fill[fill=green!40!white] (1,-0.1) rectangle (3,0.1);
\fill[fill=red!40!white] (-0.1,1) rectangle (0.1,3);
\draw[thick] (2,4) -- (4,2);
\draw (1pt,3 cm) -- (-1pt,3 cm) node[anchor=east] {$y$};
\draw (3 cm,1pt) -- (3 cm,-1pt) node[anchor=north] {$x$};
\draw [dashed, decoration={markings,mark=at position 1 with
    {\arrow[scale=2,>=latex]{>}}},postaction={decorate}] (3,3) -- (0,3);
    \draw [dashed, decoration={markings,mark=at position 1 with
    {\arrow[scale=2,>=latex]{>}}},postaction={decorate}] (3,0) -- (3,3);
    \filldraw[fill=black, draw=black] (0,3) circle (.5mm);
\filldraw[fill=black, draw=black] (3,0) circle (.5mm);
\end{tikzpicture}
\end{center}
\caption{\footnotesize{\textbf{(left)} Depiction of complementary slackness. The horizontal axis represents a vector space in which we consider a cone program (e.g., a linear or semidefinite program), and the feasibility region of this program is highlighted in green. The dual program concerns another vector space, which we represent with the vertical axis and feasibility region highlighted in red. The downward-sloping line represents all pairs of points $(x,y)$ that satisfy complementary slackness. Recall that when strong duality is satisfied, we have that $x$ is primal-optimal and $y$ is dual-optimal if and only if $x$ is primal feasible, $y$ is dual feasible, and $(x,y)$ satisfy complementary slackness. As such, the intersection between the blue Cartesian product and the complementary slackness line represents all pairs of optimizers. \textbf{(right)} Bandeira's probably certifiably correct technique~\cite{bandeira2015note}. Given a purported primal-optimizer $x$, we first check that $x$ is primal-feasible. Next, we select $y$ such that $(x,y)$ satisfies complementary slackness. Finally, we check that $y$ is dual-feasible. By complementary slackness, $y$ is then a dual certificate of $x$'s optimality in the primal program, which can be verified by comparing their values (a la strong duality).}\label{figure.pcc}}
\end{figure}

\subsection{Technical background and overview}

The first two components of a probably certifiably correct algorithm require non-convex and convex solvers that perform well under some ``reasonable'' distribution of problem instances.
In the context of geometric clustering, it has become popular recently to consider a particular model of data called the \textbf{stochastic ball model}, introduced in~\cite{Nellore_Kmedians}:

\begin{definition}[$(\mathcal D, \gamma, n)$-stochastic ball model]
\label{stochastic_balls}
Let $\{\gamma_a\}_{a=1}^k$ be ball centers in $\mathbb{R}^m$.
For each $a$, draw i.i.d.\ vectors $\{r_{a,i}\}_{i=1}^n$ from some rotation-invariant distribution $\mathcal D$ supported on the unit ball.
The points from cluster $a$ are then taken to be $x_{a,i}:= r_{a,i} + \gamma_a$.
We denote $\Delta:=\min_{a\neq b}\|\gamma_a-\gamma_b\|_2$. 
\end{definition}

Table~\ref{table} summarizes the state of the art for recovery guarantees under the stochastic ball model.
In~\cite{Nellore_Kmedians}, it was shown that an LP relaxation of $k$-medians will, with high probability, recover clusters drawn from the stochastic ball model provided the smallest distance between ball centers is $\Delta\geq3.75$.
Note that exact recovery only makes sense for $\Delta>2$ (i.e., when the balls are disjoint).
Once $\Delta>4$, any two points within a particular cluster are closer to each other than any two points from different clusters, and so in this regime, cluster recovery follows from a simple distance thresholding.
For the $k$-means problem, Awasthi et al.~\cite{relax} studies the Peng--Wei semidefinite relaxation and demonstrates exact recovery in the regime $\Delta>2\sqrt2(1+ 1/\sqrt m)$, where $m$ is the dimension of the Euclidean space.

\begin{table}
\begin{center}
\begin{tabular}{llll} \hline
    {\textbf{Method}} & {\textbf{Sufficient Condition}} & {\textbf{Optimal?}} & {\textbf{Reference}} \T\B \\ \hline\hline
    Thresholding  & $\Delta>4$ & Yes & (simple exercise) \T\B \\ \hline
    $k$-medians LP  & $\Delta\geq4$  & No & Theorem~2 in \cite{elhamifar2012finding} \T \\
      & $\Delta\geq3.75$  & No & Theorem~1 in \cite{Nellore_Kmedians} \\
      & $\Delta>2$  & Yes &  Theorem~1 in \cite{relax} \B \\ \hline
    $k$-means LP  & $\Delta>4$   & Yes & Theorem~9 in \cite{relax} \T \B \\ \hline
    $k$-means SDP  & $\Delta>2\sqrt2(1+ 1/\sqrt m)$  & No  & Theorem~3 in \cite{relax} \T\\
      & $\Delta>2+k^2/m$  & No  & Theorem~\ref{main_theorem} \B   \\ \hline
        Spectral $k$-means  & $\Delta>\Delta^\star$, $k=2$   & Yes & Theorem~\ref{theorem.spectral clustering} \T \B \\ \hline
\end{tabular}
\end{center}
\caption{\label{table}
{\footnotesize Summary of cluster recovery guarantees under the stochastic ball model.
The second column reports sufficient separation between ball centers in order for the corresponding method to provably give exact recovery with high probability.
The third column reports whether the sufficient condition on $\Delta$ cannot be improved.
Here, $\Delta^\star=\Delta^\star(\mathcal{D},k)$ denotes the smallest value for which $\Delta>\Delta^\star$ implies that minimizing the $k$-means objective recovers planted clusters under the $(\mathcal{D},\gamma,n)$-stochastic ball model with probability $1-e^{-\Omega_{\mathcal{D},\gamma}(n)}$.}
\normalsize}
\end{table}

As indicated in Table~\ref{table}, we also study the Peng--Wei SDP, but our guarantee is different from~\cite{relax}.
In particular, we demonstrate tightness in the regime $\Delta>2+k^2/m$, which is near-optimal for large $m$.
The source of this improvement is a different choice of dual certificate, which leads to the following result (see Section \ref{sec:sdp} for details):

\begin{theorem}
\label{thm.dual certificate intro}
Take $X$ of the form \eqref{eq.integral}, and let $P_{\Lambda^\perp}$ denote the orthogonal projection onto the orthogonal complement of the span of $\{1_{A_t}\}_{t=1}^k$.
Then there exists an explicit matrix $Z=Z(D,X)$ and scalar $z=z(D,X)$ such that $X$ is a solution to the semidefinite relaxation \eqref{eq.kmeansSDP} if
\begin{equation}
\label{cert_condition_intro}
P_{\Lambda^\perp}ZP_{\Lambda^\perp}\preceq zP_{\Lambda^\perp}.
\end{equation}
\end{theorem}

To prove that $\Delta>2+k^2/m$ suffices for the SDP to recover the planted clustering under the stochastic ball model, we estimate the left- and right-hand sides of \eqref{cert_condition_intro} with the help of standard techniques from random matrix theory and concentration of measure; see Appendix~\ref{sec:appendix_theorem} for the (rather technical) details.
While this is an improvement over the condition from~\cite{relax} in the large-$m$ regime, there are other regimes (e.g., $k=m$) for which their condition is much better, leaving open the question of what the optimal bound is.
Conjecture~4 in~\cite{relax} suggests that $\Delta>2$ suffices for the $k$-means SDP to recover planted clusters under the stochastic ball model, but as we illustrate in Section \ref{sec:SBM_SDP}, this conjecture is surprisingly false.

Having accomplished component (ii) in Bandeira's PCC technique, we tackle component (iii) next.
For this, consider the matrix 
\begin{equation}
\label{eq.A matrix intro}
A:=\frac{z}{N}11^\top+P_{\Lambda^\perp}ZP_{\Lambda^\perp},
\end{equation}
where $z$ and $Z$ come from Theorem~\ref{thm.dual certificate intro}.
Since the all-ones vector $1$ lies in the span of $\{1_{A_t}\}_{t=1}^k$, we have that $1$ spans the unique leading eigenspace of $A$ precisely when $P_{\Lambda^\perp}ZP_{\Lambda^\perp}\prec zP_{\Lambda^\perp}$, which in turn implies that $X$ is a $k$-means optimal clustering by Theorem~\ref{thm.dual certificate intro}.
As such, component (iii) can be accomplished by solving the following fundamental problem from linear algebra:

\begin{problem}
\label{prob.leading espace}
Given a symmetric matrix $A\in\mathbb{R}^{n\times n}$ and an eigenvector $v$ of $A$, determine whether the span of $v$ is the unique leading eigenspace, that is, the corresponding eigenvalue $\lambda$ has multiplicity~$1$ and satisfies $|\lambda|>|\lambda'|$ for every other eigenvalue $\lambda'$ of $A$.
\end{problem}

Interestingly, this same problem appeared in Bandeira's original PCC theory~\cite{bandeira2015note}, but it was left unresolved.
In this paper, we fill this gap by developing a so-called power iteration detector, which applies the power iteration to a random initialization on the unit sphere.
Due to the randomness, the power iteration produces a test statistic that allows us to infer whether $(A,v)$ satisfies the desired leading eigenspace condition.
In Section \ref{sec:cert}, we pose this as a hypothesis test, and we estimate the associated error probabilities.
In addition, we show how to leverage the structure of $A$ defined by \eqref{eq.A matrix intro} and Theorem~\ref{cert_condition_intro} to compute the matrix--vector multiplication $Ax$ for any given $x$ in only $O(kmN)$ operations, thereby allowing the test statistic to be computed in linear time (up to the spectral gap of $A$ and the desired confidence for the hypothesis test).
See Figure \ref{figure.2} for an illustration of the runtime of our method.
Overall, the power iteration detector will deliver a highly confident inference on whether $(A,v)$ satisfies the leading eigenspace condition, which in turn certifies the optimality of $X$ up to the prescribed confidence level.
Of course, one may remove the need for a confidence level by opting for deterministic spectral methods, but we have no idea how to accomplish this in linear or even near-linear time.

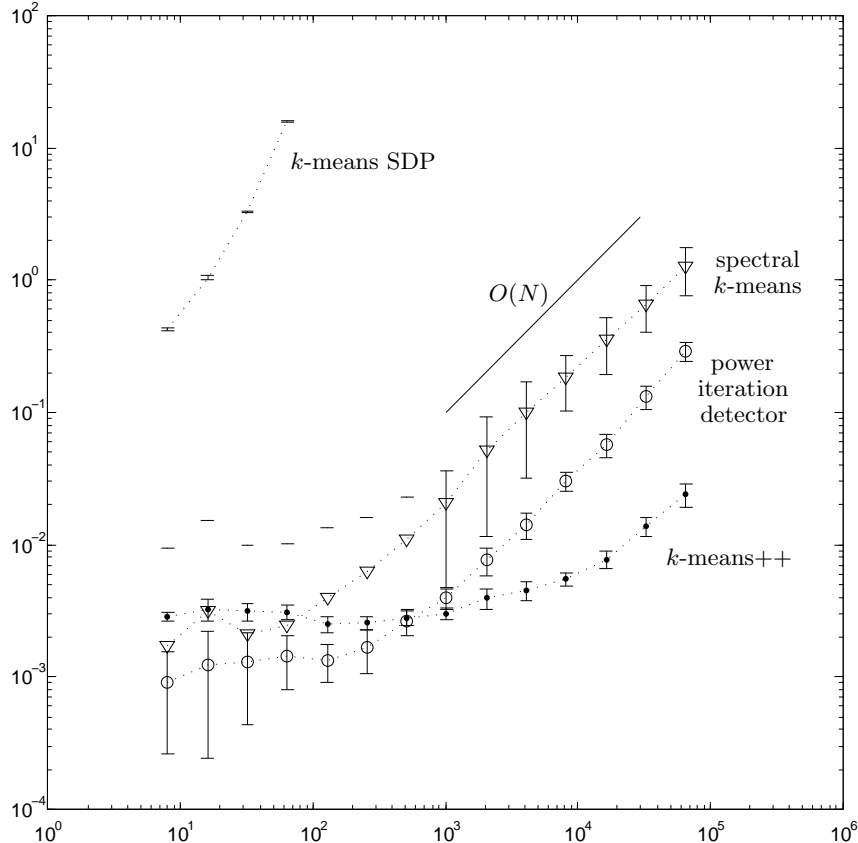
\begin{figure}[t!]
\begin{center}
\begin{picture}(330,330)
\put(0,0){\includegraphics[width=0.7\textwidth]{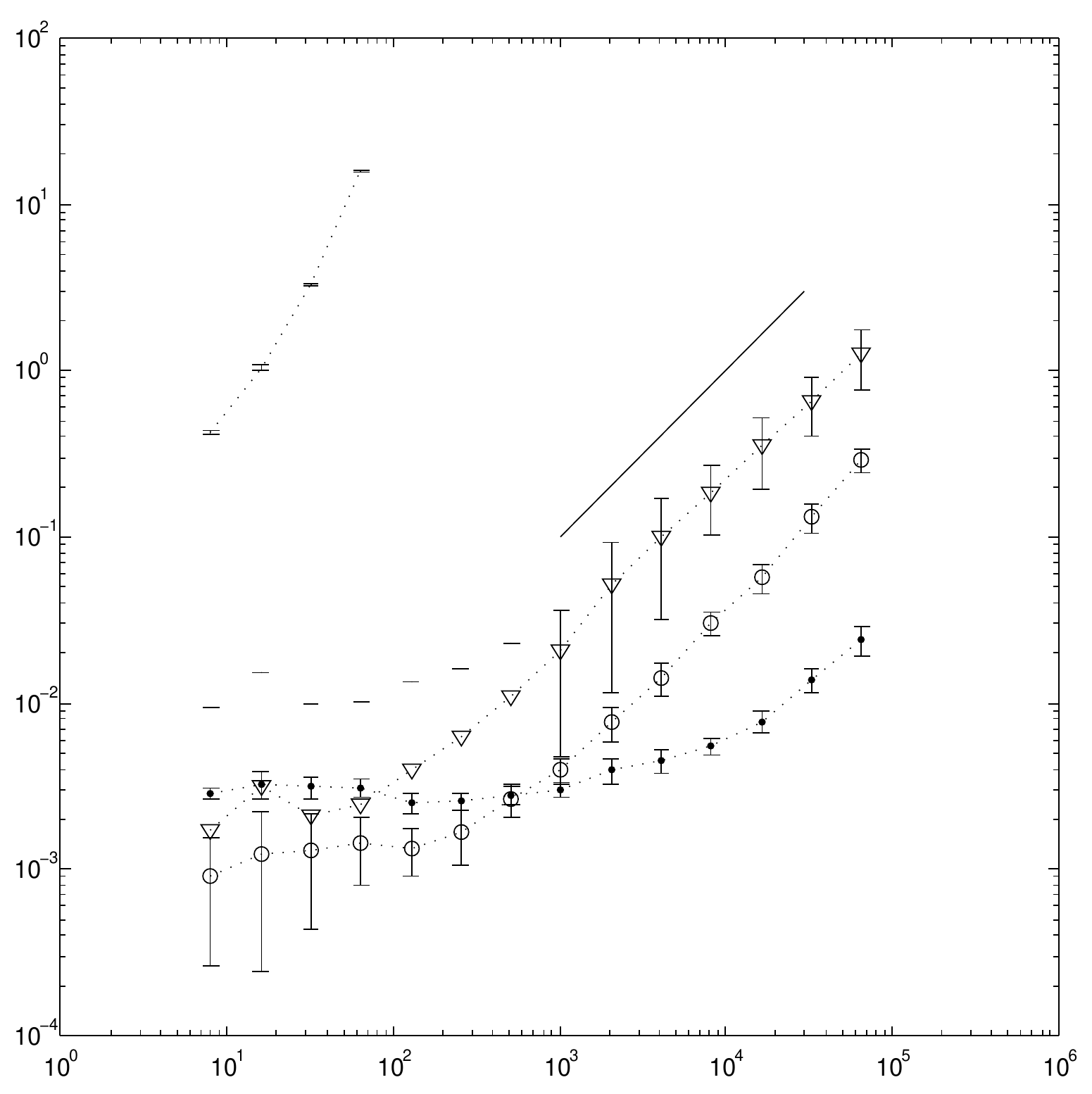}}
\put(185,210){\footnotesize{$O(N)$}}
\put(110,260){\footnotesize{$k$-means SDP}}
\put(271,223){\footnotesize{spectral}}
\put(270,213){\footnotesize{$k$-means}}
\put(269,184){\footnotesize{power}}
\put(264,174){\footnotesize{iteration}}
\put(265,164){\footnotesize{detector}}
\put(252,110){\footnotesize{$k$-means++}}
\end{picture}
\end{center}
\caption{
\label{figure.2}
{\footnotesize 
Take two unit balls in $\mathbb{R}^6$ at distance $2.3$ apart.
For each $N\in\{2^3,2^4,\ldots,2^{16}\}$, we perform $300$ trials of the following experiment:
Draw $N/2$ points uniformly at random from each ball, and then compute four different functions: (a) MATLAB's built-in implementation of $k$-means++, (b) a CVX implementation~\cite{grant2008cvx} of the $k$-means SDP~\eqref{eq.kmeansSDP}, (c) the power iteration detector (Algorithm~\ref{alg.power iteration}) with $A$ given by~\eqref{eq.A matrix intro}, and (d) spectral $k$-means clustering (Algorithm~\ref{alg.euclidean spectral clustering}).
For each trial, we recorded the runtime in seconds.
Above, we plot the average runtime along with error bars for standard deviation.
For the record, the power iteration detector failed to certify optimality (i.e., reject $H_0$ in~\eqref{eq.hypotheses}) in at most 3\% of the trials with $N\leq2^7$, but rejected $H_0$ in every trial otherwise; similarly, spectral $k$-means failed to recover the planted clusters in two of the trials with $N=2^3$.
Our implementation of the $k$-means SDP was too slow to perform trials with $N\geq 2^7$ in a reasonable amount of time, so we only recorded runtimes for $N\in\{2^3,2^4,2^5,2^6\}$.
As the plot illustrates, the other algorithms ran in quasilinear time, as expected.}
\normalsize}
\end{figure}

Now that we have discussed components (ii) and (iii) in Bandeira's PCC technique, we conclude by discussing component (i).
While we presume that there exists a fast initialization of Lloyd's algorithm that performs well under the stochastic ball model, we leave this investigation for future research.
Instead, Section \ref{sec:fast_solver} considers a spectral method introduced by Peng and Wei in \cite{peng2007approximating}.
We show that when $k=2$, this method performs as well as the optimizer of the original $k$-means problem under the stochastic ball model.
Figure \ref{figure.2} illustrates the quasilinear runtime of this approach.

\subsection{Outline}

In this paper, we provide a theoretical analysis of probably certifiably correct $k$-means clustering, and we do so by developing components (i), (ii) and (iii) of Bandeira's general technique.
First, we investigate (ii) in Section \ref{sec:sdp} by analyzing the tightness of the Peng--Wei SDP.
In particular, we choose a different dual certificate from the one used in \cite{relax}, and our choice demonstrates tightness in the SDP for clusters that are near-optimally close.
Section \ref{sec:cert} then addresses (iii) by providing a fast method of computing this dual certificate given the optimal $k$-means partition.
In fact, a subroutine of our method (the so-called power iteration detector) resolves a gap in Bandeira's original PCC theory~\cite{bandeira2015note}, and as such, we expect this to be leveraged in future PCC algorithms.
We conclude in Section \ref{sec:fast_solver} with some theoretical guarantees for (i).
Here, we focus on the case $k=2$, and we show that a slight modification of the spectral clustering--based method in \cite{peng2007approximating} manages to recover the optimal $k$-means partition with high probability under the stochastic ball model.
We conclude in Section \ref{sec:discussion} with a discussion of various open problems.

\section{A typically tight relaxation of $k$-means}
\label{sec:sdp}

This section establishes that the Peng--Wei semidefinite relaxation \eqref{eq.kmeansSDP} of the $k$-means problem \eqref{kmeans} is typically tight under the stochastic ball model.
First, we find a deterministic condition on the set of points under which the relaxation finds the $k$-means-optimal solution.
Later, we discuss when this deterministic condition is satisfied with high probability under the stochastic ball model.

\subsection{The dual program} \label{sec:sdp_intro}
The following is the dual program of \eqref{eq.kmeansSDP}:
\begin{alignat}{2}
\label{eq.kmeansSDPdual}
& \text{minimize}  &       & kz+\sum_{i=1}^N\alpha_i \\
\nonumber
& \text{subject to}& \quad & 
\begin{aligned}[t]
Q:=zI+\sum_{i=1}^N\alpha_i\cdot\frac{1}{2}(e_i1^\top+1e_i^\top)-\sum_{i=1}^N\sum_{j=i}^N\beta_{ij}\cdot\frac{1}{2}(e_ie_j^\top+e_je_i^\top)+D&\succeq0\\
\beta&\geq0
\end{aligned}
\end{alignat}

For notational simplicity, from this point forward, we organize indices according to clusters.
For example, $1_a$ shall denote the indicator function of the $a$th cluster.
Also, we shuffle the rows and columns of $X$ and $D$ into blocks that correspond to clusters; for example, the $(i,j)$th entry of the $(a,b)$th block of $D$ is given by $D^{(a,b)}_{ij}$.
We also index $\alpha$ in terms of clusters; for example, the $i$th entry of the $a$th block of $\alpha$ is denoted $\alpha_{a,i}$.
For $\beta$, we identify
\[
\beta:=\sum_{i=1}^N\sum_{j=i}^N\beta_{ij}\cdot\frac{1}{2}(e_ie_j^\top+e_je_i^\top).
\]
Indeed, when $i\leq j$, the $(i,j)$th entry of $\beta$ is $\beta_{ij}$.
We also consider $\beta$ as having its rows and columns shuffled according to clusters, so that the $(i,j)$th entry of the $(a,b)$th block is $\beta^{(a,b)}_{ij}$.

With this notation, the following proposition characterizes all possible dual certificates of \eqref{eq.kmeansSDP}:

\begin{proposition}[Theorem~4 in~\cite{iguchi2015tightness}, cf.\ \cite{relax}]
\label{thm.integral optimality}
Take $X:=\sum_{a=1}^k\frac{1}{n_a}1_a1_a^\top$, where $n_a$ denotes the number of points in cluster $t$.
The following are equivalent:
\begin{itemize}
\item[(a)]
$X$ is a solution to the semidefinite relaxation \eqref{eq.kmeansSDP}.
\item[(b)]
Every solution to the dual program \eqref{eq.kmeansSDPdual} satisfies
\[
Q^{(a,a)}1=0,
\qquad
\beta^{(a,a)}=0
\qquad
\forall a\in\{1,\ldots,k\}.
\]
\item[(c)]
Every solution to the dual program \eqref{eq.kmeansSDPdual} satisfies
\[
\alpha_{a,r}
=-\frac{1}{n_a}z+\frac{1}{n_a^2}1^\top D^{(a,a)}1-\frac{2}{n_a}e_r^\top D^{(a,a)}1
\qquad
\forall a\in\{1,\ldots,k\},~r\in a.
\]
\end{itemize}
\end{proposition}

The following subsection will leverage this result to identify a condition on $D$ that implies that the SDP \eqref{eq.kmeansSDP} relaxation is tight.

\subsection{Selecting a dual certificate} \label{sec:deterministic}

The goal is to certify when the SDP relaxation is tight.
In this event, Proposition~\ref{thm.integral optimality} characterizes acceptable dual certificates $(z,\alpha,\beta)$, but this information fails to uniquely determine a certificate.
In this subsection, we will motivate the application of additional constraints on dual certificates so as to identify certifiable instances.

We start by reviewing the characterization of dual certificates $(z,\alpha,\beta)$ provided in Proposition~\ref{thm.integral optimality}.
In particular, $\alpha$ is completely determined by $z$, and so $z$ and $\beta$ are the only remaining free variables.
Indeed, for every $a,b\in\{1,\ldots,k\}$, we have
\begin{align*}
&\bigg(\sum_{t=1}^k\sum_{i\in t}\alpha_{t,i}\cdot\frac{1}{2}(e_{t,i}1^\top+1e_{t,i}^\top)\bigg)^{(a,b)}\\
&\qquad=\sum_{i\in a}\alpha_{a,i}\cdot\frac{1}{2}e_i1^\top+\sum_{j\in b}\alpha_{b,j}\cdot\frac{1}{2}1e_j^\top\\
&\qquad=-\frac{1}{2}\bigg(\frac{1}{n_a}+\frac{1}{n_b}\bigg)z+\sum_{i\in a}\bigg(\frac{1}{n_a^2}1^\top D^{(a,a)}1-\frac{2}{n_a}e_i^\top D^{(a,a)}1\bigg)\frac{1}{2}e_i1^\top\\
&\qquad\qquad+\sum_{j\in b}\bigg(\frac{1}{n_b^2}1^\top D^{(b,b)}1-\frac{2}{n_b}e_j^\top D^{(b,b)}1\bigg)\frac{1}{2}1e_j^\top,
\end{align*}
and so since
\[
Q
=zI+\sum_{t=1}^k\sum_{i\in t}\alpha_{t,i}\cdot\frac{1}{2}(e_{t,i}1^\top+1e_{t,i}^\top)-\frac{1}{2}\beta+D,
\]
we may write $Q=z(I-E)+M-B$, where
\begin{align}
\label{eq.definition of E}
E^{(a,b)}
&:=\frac{1}{2}\bigg(\frac{1}{n_a}+\frac{1}{n_b}\bigg)11^\top\\
\nonumber
M^{(a,b)}
&:=D^{(a,b)}+\sum_{i\in a}\bigg(\frac{1}{n_a^2}1^\top D^{(a,a)}1-\frac{2}{n_a}e_i^\top D^{(a,a)}1\bigg)\frac{1}{2}e_i1^\top\\
\label{eq.definition of M}
&\qquad+\sum_{j\in b}\bigg(\frac{1}{n_b^2}1^\top D^{(b,b)}1-\frac{2}{n_b}e_j^\top D^{(b,b)}1\bigg)\frac{1}{2}1e_j^\top\\
\nonumber
B^{(a,b)}
&=\frac{1}{2}\beta^{(a,b)}
\end{align}
for every $a,b\in\{1,\ldots,k\}$.
The following is one way to formulate our task:
Given $D$ and a clustering $X$ (which in turn determines $E$ and $M$), determine whether there exist feasible $z$ and $B$ such that $Q\succeq0$; here, feasibility only requires $B$ to be symmetric with nonnegative entries and $B^{(a,a)}=0$ for every $a\in\{1,\ldots,k\}$.
We opt for a slightly more modest goal:
Find $z=z(D,X)$ and $B=B(D,X)$ such that $Q\succeq0$ for a large family of $D$'s.

Before determining $z$ and $B$, we first analyze $E$:

\begin{lemma}
\label{lemma.eigen of E}
Let $E$ be the matrix defined by \eqref{eq.definition of E}.
Then $\operatorname{rank}(E)\in\{1,2\}$.
The eigenvalue of largest magnitude is $\lambda\geq k$, and when $\operatorname{rank}(E)=2$, the other nonzero eigenvalue of $E$ is negative.
The eigenvectors corresponding to nonzero eigenvalues lie in the span of $\{1_a\}_{a=1}^k$.
\end{lemma}

\begin{proof}
Writing
\begin{equation*}
E
=\sum_{a=1}^k\sum_{b=1}^k\frac{1}{2}\bigg(\frac{1}{n_a}+\frac{1}{n_b}\bigg)1_a1_b^\top
=\frac{1}{2}\bigg(\sum_{a=1}^k\frac{1}{n_a}1_a\bigg)1^\top+\frac{1}{2}1\bigg(\sum_{b=1}^k\frac{1}{n_b}1_b\bigg)^\top,
\end{equation*}
we see that $\operatorname{rank}(E)\in\{1,2\}$, and it is easy to calculate $1^\top E1=Nk$ and $\operatorname{Tr}(E)=k$.
Observe that
\[
\lambda
=\sup_{\substack{x\in\mathbb{R}^N\\\|x\|_2=1}}x^\top Ex
\geq \frac{1}{N}1^\top E1
=k,
\]
and combining with $\operatorname{rank}(E)\leq2$ and $\operatorname{Tr}(E)=k$ then implies that the other nonzero eigenvalue (if there is one) is negative.
Finally, any eigenvector of $E$ with a nonzero eigenvalue necessarily lies in the column space of $E$, which is a subspace of $\operatorname{span}\{1_a\}_{a=1}^k$ by the definition of $E$.
\end{proof}

When finding $z$ and $B$ such that $Q=z(I-E)+M-B\succeq0$, it will be useful that $I-E$ has only one negative eigenvalue to correct.
Let $v$ denote the corresponding eigenvector.
Then we will pick $B$ so that $v$ is also an eigenvector of $M-B$.
Since we want $Q\succeq0$ for as many instances of $D$ as possible, we will then pick $z$ as large as possible, thereby sending $v$ to the nullspace of $Q$.
Unfortunately, the authors found that this constraint fails to uniquely determine $B$ in general.
Instead, we impose a stronger constraint:
\[
Q1_a=0
\qquad
\forall a\in\{1,\ldots,k\}.
\]
(This constraint implies $Qv=0$ by Lemma~\ref{lemma.eigen of E}.)
To see the implications of this constraint, note that we already necessarily have
\[
(Q1_a)_a
=\Big((z(I-E)+M-B)1_a\Big)_a
=z(I-E^{(a,a)})1+M^{(a,a)}1-B^{(a,a)}1
=z\bigg(1-\frac{1}{n_a}11^\top1\bigg)
=0,
\]
and so it remains to impose
\begin{align}
\nonumber
0
=(Q1_b)_a
&=\Big((z(I-E)+M-B)1_b\Big)_a\\
\label{eq.requirement for B}
&=-zE^{(a,b)}1+M^{(a,b)}1-B^{(a,b)}1
=-z\frac{n_a+n_b}{2n_a}1+M^{(a,b)}1-B^{(a,b)}1.
\end{align}
In order for there to exist a vector $B^{(a,b)}1\geq0$ that satisfies \eqref{eq.requirement for B}, $z$ must satisfy
\[
z\frac{n_a+n_b}{2n_a}\leq\min(M^{(a,b)}1),
\]
and since $z$ is independent of $(a,b)$, we conclude that
\begin{equation}
\label{eq.bound on z}
z\leq\min_{\substack{a,b\in\{1,\ldots,k\}\\a\neq b}}\frac{2n_a}{n_a+n_b}\min(M^{(a,b)}1).
\end{equation}
Again, in order to ensure $z(I-E)+M-B\succeq0$ for as many instances of $D$ as possible, we intend to choose $z$ as large as possible.
Luckily, there is a choice of $B$ which satisfies \eqref{eq.requirement for B} for every $(a,b)$, even when $z$ satisfies equality in \eqref{eq.bound on z}.
Indeed, we define
\begin{equation}
\label{eq.how to construct B}
u_{(a,b)}
:=M^{(a,b)}1-z\frac{n_a+n_b}{2n_a}1,
\qquad
\rho_{(a,b)}
:=u_{(a,b)}^\top 1,
\qquad
B^{(a,b)}
:=\frac{1}{\rho_{(b,a)}}u_{(a,b)}u_{(b,a)}^\top
\end{equation}
for every $a,b\in\{1,\ldots,k\}$ with $a\neq b$.
Then by design, $B$ immediately satisfies \eqref{eq.requirement for B}.
Also, note that $\rho_{(a,b)}=\rho_{(b,a)}$, and so $B^{(b,a)}=(B^{(a,b)})^\top$, meaning $B$ is symmetric.
Finally, we necessarily have $u_{(a,b)}\geq0$ (and thus $\rho_{(a,b)}\geq0$) by \eqref{eq.bound on z}, and we implicitly require $\rho_{(a,b)}>0$ for division to be permissible.
As such, we also have $B^{(a,b)}\geq0$, as desired.

Now that we have selected $z$ and $B$, it remains to check that $Q\succeq0$.
By construction, we already have $\Lambda:=\operatorname{span}\{1_a\}_{a=1}^k$ in the nullspace of $Q$, and so it suffices to ensure
\[
0
\preceq P_{\Lambda^\perp}QP_{\Lambda^\perp}
=P_{\Lambda^\perp}\Big(z(I-E)+M-B\Big)P_{\Lambda^\perp}
=zP_{\Lambda^\perp}+P_{\Lambda^\perp}(M-B)P_{\Lambda^\perp}.
\]
Here, $P_{\Lambda^\perp}$ denotes the orthogonal projection onto the orthogonal complement of $\Lambda$.
Rearranging then gives the following result:

\begin{theorem}
\label{thm.dual certificate}
Take $X:=\sum_{t=1}^k\frac{1}{n_t}1_t1_t^\top$, where $n_t$ denotes the number of points in cluster $t$.
Consider $M$ defined by \eqref{eq.definition of M}, pick $z$ so as to satisfy equality in \eqref{eq.bound on z}, take $B$ defined by \eqref{eq.how to construct B}, and let $\Lambda$ denote the span of $\{1_t\}_{t=1}^k$.
Then $X$ is a solution to the semidefinite relaxation \eqref{eq.kmeansSDP} if
\begin{equation}
\label{cert_condition}
P_{\Lambda^\perp}(B-M)P_{\Lambda^\perp}\preceq zP_{\Lambda^\perp}.
\end{equation}
\end{theorem}

The next subsection leverages this sufficient condition to establish that the Peng--Wei SDP \eqref{eq.kmeansSDP} is typically tight under the stochastic ball model.

\subsection{Integrality of the relaxation under the stochastic ball model} \label{sec:SBM_SDP}

We first note that our sufficient condition \eqref{cert_condition} is implied by
\[
\|P_{\Lambda^\perp}MP_{\Lambda^\perp}\|_{2\rightarrow2}+\|P_{\Lambda^\perp}BP_{\Lambda^\perp}\|_{2\rightarrow2}\leq z.
\]
By further analyzing the left-hand side above (see Appendix~\ref{sec:appendix_corollary}), we arrive at the following corollary:

\begin{corollary}
\label{cor.dual certificate}
Take $X:=\sum_{t=1}^k\frac{1}{n_t}1_t1_t^\top$, where $n_t$ denotes the number of points in cluster $t$.
Let $\Psi$ denote the $m\times N$ matrix whose $(a,i)$th column is $x_{a,i}-c_a$, where
\[
c_a:=\frac{1}{n_a}\sum_{i\in a}x_{a,i}
\]
denotes the empirical center of cluster $a$.
Consider $M$ defined by \eqref{eq.definition of M}, pick $z$ so as to satisfy equality in \eqref{eq.bound on z}, and take $\rho_{(a,b)}$ defined by \eqref{eq.how to construct B}.
Then $X$ is a solution to the semidefinite relaxation \eqref{eq.kmeansSDP} if
\[
2\|\Psi\|_{2\rightarrow2}^2+\sum_{a=1}^k\sum_{b=a+1}^k\frac{\|P_{1^\perp}M^{(a,b)}1\|_2\|P_{1^\perp}M^{(b,a)}1\|_2}{\rho_{(a,b)}}
\leq z.
\]
\end{corollary}

In Appendix~\ref{sec:appendix_theorem}, we leverage the stochastic ball model to bound each term in Corollary~\ref{cor.dual certificate}, and in doing so, we identify a regime in which the data points typically satisfy the sufficient condition given in Corollary~\ref{cor.dual certificate}:

\begin{theorem}
\label{main_theorem}
The $k$-means semidefinite relaxation \eqref{eq.kmeansSDP} recovers the planted clusters in the $(\mathcal D, \gamma,n)$-stochastic ball model with probability $1-e^{-\Omega_{\mathcal D, \gamma}(n)}$ provided $\Delta>2+k^2/m$.
\end{theorem}

We note that Theorem~\ref{main_theorem} is an improvement to the main result of the authors' preprint~\cite{iguchi2015tightness}.
When $k=o(m^{1/2})$, Theorem~\ref{main_theorem} is near-optimal, and in this sense, it's a significant improvement over the sufficient condition
\begin{equation}
\label{eq.bound from relax}
\Delta>2\sqrt{2}\bigg(1+\frac{1}{\sqrt{m}}\bigg)
\end{equation}
given in \cite{relax}.
However, there are regimes (e.g., $k=m$) for which \eqref{eq.bound from relax} is much better, leaving open the question of what the optimal bound is.
Conjecture~4 in~\cite{relax} suggests that $\Delta>2$ suffices for the $k$-means SDP to recover planted clusters under the stochastic ball model, but as we illustrate below, this conjecture is surprisingly false.

Consider the special case where $m=1$, $\mathcal{D}$ is uniform on $\{\pm1\}$, and $k=2$.
Centering the two balls on $\pm\Delta/2$, then all of the points land in $\{\pm\Delta/2\pm1\}$.
The $k$-means-optimal clustering will partition the real line into two semi-infinite intervals, and so there are three possible ways of clustering these points.
Suppose exactly $N/4$ of the points land in each of the $4$ positions.
Then by symmetry, there are only two ways to cluster: either we select the planted clusters, or we make the left-most location its own cluster.
Interestingly, a little algebra reveals that this second alternative is strictly better in the $k$-means sense provided $\Delta<1+\sqrt{3}\approx 2.7320$.
Also, in this regime, then as $N$ gets large, the proportion of points in each position will be so close to $1/4$ (with high probability) that this clustering will beat the planted clusters.

Overall, when $m=1$ and $k=2$, then $\Delta\geq1+\sqrt{3}$ is necessary for minimizing the $k$-means objective to recover planted clusters for an arbitrary $\mathcal{D}$.
As a relaxation, the $k$-means SDP recovers planted clusters only if minimizing the $k$-means objective does so as well, and so it inherits this necessary condition, thereby disproving Conjecture~4 in~\cite{relax}.
Furthermore, as Figure~\ref{figure.1}(left) illustrates, a similar counterexample is available in higher dimensions.

\begin{figure}[t]
\begin{center}
\includegraphics[width=\textwidth]{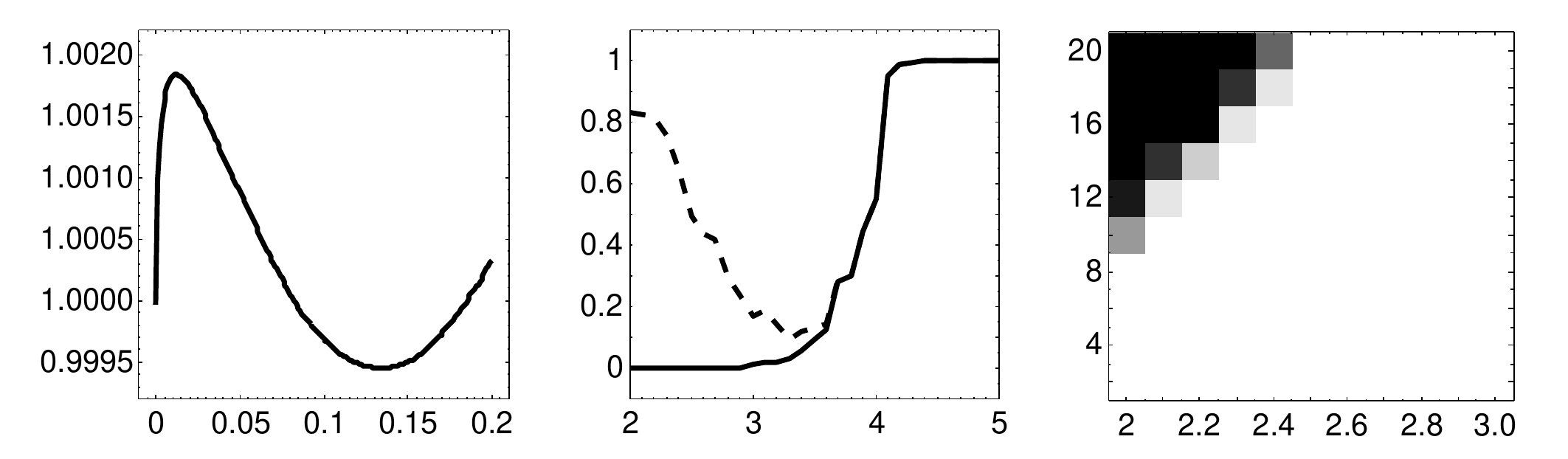}
\end{center}
\caption{
\label{figure.1}
{\footnotesize 
\textbf{(left)}
Take two unit disks in $\mathbb{R}^2$ with centers on the $x$-axis at distance $2.08$ apart.
Let $x_0$ denote the smallest possible $x$-coordinate in the disk on the right.
For each disk, draw $N/2=50,000$ points uniformly at random from the perimeter.
Given $\theta$, cluster the points according to whether the $x$-coordinate is smaller than $x_0+\theta$.
When $\theta=0$, this clustering gives the planted clusters, and the $k$-means objective (divided by $N$) is $1$.
We plot this normalized $k$-means objective for $\theta\in[0,0.2]$.
Since $N$ is large, this curve is very close to its expected shape, and we see that there are clusters whose $k$-means value is smaller than that of the planted clustering.
\textbf{(center)}
Take two intervals of width $2$ in $\mathbb{R}$, and let $\Delta$ denote the distance between the midpoints of these intervals.
For each interval, draw $10$ points at random from its endpoints, and then run the $k$-means SDP.
For each $\Delta=2:0.1:5$, after running $2,000$ trials of this experiment, we plot the proportion of trials for which the SDP relaxation was tight (dashed line) and the proportion of trials for which the SDP recovered the planted clusters (solid line).
In this case, cluster recovery appears to exhibit a phase transition at $\Delta=4$.
\textbf{(right)}
For each $\Delta=2:0.1:3$ and $k=2:2:20$, consider the unit balls in $\mathbb{R}^{20}$ centered at $\{\frac{\Delta}{\sqrt{2}}e_i\}_{i=1}^k$, where $e_i$ denotes the $i$th identity basis element.
Draw $100$ points uniformly from each ball, and test if the resulting data points satisfy \eqref{cert_condition}.
After performing $10$ trials of this experiment for each $(\Delta,k)$, we shade the corresponding pixel according to the proportion of successful trials (white means every trial satisfied \eqref{cert_condition}).
This plot indicates that our certificate \eqref{cert_condition} is to blame for Theorem~\ref{main_theorem}'s dependence on $k$.}
\normalsize}
\end{figure}

To study when the SDP recovers the clusters, let's continue with the case where $m=1$ and $k=2$.
We know that minimizing $k$-means will recover the clusters with high probability provided $\Delta>1+\sqrt{3}$.
However, Theorem~\ref{main_theorem} only guarantees that the SDP recovers the clusters when $\Delta>6$; in fact, \eqref{eq.bound from relax} is slightly better here, giving that $\Delta\geq5.6569$ suffices.
To shed light on the disparity, Figure~\ref{figure.1}(center) illustrates the performance of the SDP for different values of $\Delta$.
Observe that the SDP is often tight when $\Delta$ is close to $2$, but it doesn't reliably recover the planted clusters until $\Delta>4$.
We suspect that $\Delta=4$ is a phase transition for cluster recovery in this case.

Qualitatively, the biggest difference between Theorem~\ref{main_theorem} and \eqref{eq.bound from relax} is the dependence on $k$ that Theorem~\ref{main_theorem} exhibits.
Figure~\ref{figure.1}(right) illustrates that this comes from \eqref{cert_condition}, meaning that one would need to use a completely different dual certificate in order to remove this dependence.

\section{A fast test for $k$-means optimality}
\label{sec:cert}

In this section, we leverage the certificate~\eqref{cert_condition} to test the optimality of a candidate $k$-means solution.
We first show how to solve a more general problem from linear algebra, and then we apply our solution to devise a fast test for $k$-means optimality (as well as fast test for a related PCC algorithm).

\subsection{Leading eigenvector hypothesis test}\label{sec:evec}

This subsection concerns Problem~\ref{prob.leading espace}.
To solve this problem, one might be inclined to apply the power method:

\begin{proposition}[Theorem~8.2.1 in~\cite{golub2012matrix}]
\label{prop.power iteration guarantee}
Let $A\in\mathbb{R}^{n\times n}$ be a symmetric matrix with eigenvalues $\{\lambda_i\}_{i=1}^n$ (counting multiplicities) satisfying
\[
|\lambda_1|>|\lambda_2|\geq\cdots\geq|\lambda_n|,
\]
and with corresponding orthonormal eigenvectors $\{v_i\}_{i=1}^n$.
Pick a unit-norm vector $q_0\in\mathbb{R}^n$ and consider the power iteration $q_{j+1}:=Aq_j/\|Aq_j\|_2$.
If $q_0$ is not orthogonal to $v_1$, then
\[
(v_1^\top q_j)^2
\geq1-\Big((v_1^\top q_0)^{-2}-1\Big)\bigg(\frac{\lambda_2}{\lambda_1}\bigg)^{2j}.
\]
\end{proposition}

Notice that the above convergence guarantee depends on the quality of the initialization $q_0$.
To use this guarantee, draw $q_0$ at random from the unit sphere so that $q_0$ is not orthogonal to $v_1$ almost surely; one might then analyze the statistics of $v_1^\top q_0$ to produce statistics on the time required for convergence.
The power method is typically used to find a leading eigenvector, but for our problem, we already have access to an eigenvector $v$, and we are tasked with determining whether $v$ is the unique leading eigenvector.
Intuitively, if you run the power method from a random initialization and it happens to converge to $v$, then this would have been a remarkable coincidence if $v$ were not the unique leading eigenvector.
Since we will only run finitely many iterations, how do we decide when we are sufficiently confident?
The remainder of this subsection answers this question.

Given a symmetric matrix $A\in\mathbb{R}^{n\times n}$ and a unit eigenvector $v$ of $A$, consider the hypotheses
\begin{equation}
\label{eq.hypotheses}
\begin{array}{ll}
H_0\colon&\text{$\operatorname{span}(v)$ is not the unique leading eigenspace of $A$,}\\ \T
H_1\colon&\text{$\operatorname{span}(v)$ is the unique leading eigenspace of $A$.}
\end{array}
\end{equation}
To test these hypotheses, pick a tolerance $\epsilon>0$ and run the power iteration detector (see Algorithm~\ref{alg.power iteration}).
This detector terminates either by accepting $H_0$ or by rejecting $H_0$ and accepting $H_1$.
We say the detector \textbf{fails to reject $H_0$} if it either accepts $H_0$ or fails to terminate.
Before analyzing this detector, we consider the following definition:

\begin{algorithm}[t]
\caption{Power iteration detector}
\label{alg.power iteration}
\SetAlgoLined
\KwIn{Symmetric matrix $A\in\mathbb{R}^{n\times n}$, unit eigenvector $v\in\mathbb{R}^n$, tolerance $\epsilon>0$}
\KwOut{Decision of whether to accept $H_0$ or to reject $H_0$ and accept $H_1$ as given in \eqref{eq.hypotheses}}
$\lambda\leftarrow v^\top Av$\\
Draw $q$ uniformly at random from the unit sphere in $\mathbb{R}^n$\\
\While{no decision has been made}{
\uIf{$|q^\top Aq|>|\lambda|$}{Print \texttt{accept $H_0$}}
\ElseIf{$(v^\top q)^2\geq1-\epsilon$}{Print \texttt{reject $H_0$ and accept $H_1$}}
$q\leftarrow Aq/\|Aq\|_2$
}
\end{algorithm}

\begin{definition}
Given a symmetric matrix $A\in\mathbb{R}^{n\times n}$ and unit eigenvector $v$ of $A$, put $\lambda=v^\top Av$, and let $\lambda_1$ denote a leading eigenvalue of $A$ (i.e., $|\lambda_1|=\|A\|_{2\rightarrow2}$).
We say $(A,v)$ is \textbf{degenerate} if
\begin{itemize}
\item[(a)]
the eigenvalue $\lambda$ of $A$ has multiplicity $\geq2$,
\item[(b)]
$-\lambda$ is an eigenvalue of $A$, or
\item[(c)]
$-\lambda_1$ is an eigenvalue of $A$.
\end{itemize}
\end{definition}

\begin{theorem}
Consider the power iteration detector (Algorithm~\ref{alg.power iteration}), let $q_j$ denote $q$ at the $j$th iteration (with $q_0$ being the initialization), and let $\pi_\epsilon$ denote the probability that $(e_1^\top q_0)^2<\epsilon$.
\begin{itemize}
\item[(i)]
$(A,v)$ is degenerate only if $H_0$ holds.
If $(A,v)$ is non-degenerate, then the power iteration detector terminates in finite time with probability~$1$.
\item[(ii)]
The power iteration detector incurs the following error rates:
\[
\operatorname{Pr}\Big(~\text{reject $H_0$ and accept $H_1$}~\Big|~H_0~\Big)\leq \pi_\epsilon,
\qquad
\operatorname{Pr}\Big(~\text{fail to reject $H_0$}~\Big|~H_1~\Big)=0.
\]
\item[(iii)]
If $H_1$ holds, then
\[
\min\Big\{j:(v^\top q_j)^2>1-\epsilon\Big\}
\leq\frac{3\log(1/\epsilon)}{2\log(\lambda_1/\lambda_2)}+1
\]
with probability $\geq 1-\pi_\epsilon$.
\end{itemize}
\end{theorem}

\begin{proof}
Denote the eigenvalues of $A$ by $\{\lambda_i\}_{i=1}^n$ (counting multiplicities), ordered in such a way that $|\lambda_1|\geq\cdots\geq|\lambda_n|$, and consider the corresponding orthonormal eigenvectors $\{v_i\}_{i=1}^n$, where $v=v_p$ for some $p$.

For (i), first note that $H_1$ implies that $(A,v)$ is non-degenerate, and so the contrapositive gives the first claim.
Next, suppose $(A,v)$ is non-degenerate.
If $H_1$ holds, then $(v^\top q_j)^2\rightarrow 1$ by Proposition~\ref{prop.power iteration guarantee} provided $q_0$ is not orthogonal to $v$, and so the power iteration detector terminates with probability $1$.
Otherwise, $H_0$ holds, and so the non-degeneracy of $(A,v)$ implies that the eigenspace corresponding to $\lambda_1$ is the unique leading eigenspace of $A$, and furthermore, $|\lambda_1|>|\lambda|$.
Following the proof of Theorem~8.2.1 in~\cite{golub2012matrix}, we also have
\[
q_j^\top Aq_j
=\frac{q_0^\top A^{2j+1}q_0}{q_0^\top A^{2j}q_0}
=\frac{\sum_{i=1}^n(v_i^\top q_j)^2\lambda_i^{2j+1}}{\sum_{i=1}^n(v_i^\top q_j)^2\lambda_i^{2j}}.
\]
Putting $r:=\min\{i:|\lambda_i|<|\lambda_1|\}$, then
\begin{align*}
|q_j^\top Aq_j-\lambda_1|
&=\left|\frac{\sum_{i=1}^n(v_i^\top q_j)^2\lambda_i^{2j}(\lambda_i-\lambda_1)}{\sum_{i=1}^n(v_i^\top q_j)^2\lambda_i^{2j}}\right|\\
&\leq\frac{|\lambda_1-\lambda_n|}{\|P_{\lambda_1}q_0\|_2^2}\sum_{i=r}^n(v_i^\top q_j)^2\bigg(\frac{\lambda_i}{\lambda_1}\bigg)^{2j}
\leq|\lambda_1-\lambda_n|\bigg(\frac{1-\|P_{\lambda_1}q_0\|_2^2}{\|P_{\lambda_1}q_0\|_2^2}\bigg)\bigg(\frac{\lambda_r}{\lambda_1}\bigg)^{2j},
\end{align*}
where $P_{\lambda_1}$ denotes the orthogonal projection onto the eigenspace corresponding to $\lambda_1$.
As such, $|q_j^\top Aq_j|\rightarrow|\lambda_1|>|\lambda|$ provided $P_{\lambda_1}q_0\neq0$, and so the power iteration detector terminates with probability $1$.

For (ii), we first consider the case of a false positive.
Taking $v=v_p$ for $p\neq 1$, note that $(v^\top q_j)^2>1-\epsilon$ implies
\[
\epsilon
>1-(v^\top q_j)^2
=\|q_j\|_2^2-(v_p^\top q_j)^2
=\sum_{\substack{i=1\\i\neq p}}^n(v_i^\top q_j)^2
\geq(v_1^\top q_j)^2.
\]
Also, since $\|Ax\|_2\leq|\lambda_1|\|x\|_2$ for all $x\in\mathbb{R}^n$, we have that $(v_1^\top q_j)^2$ monotonically increases with $j$:
\[
(v_1^\top q_{j+1})^2
=\bigg(v_1^\top \frac{Aq_j}{\|Aq_j\|_2}\bigg)^2
=\frac{(\lambda_1 v_1^\top q_j)^2}{\|Aq_j\|_2^2}
\geq\frac{(v_1^\top q_j)^2}{\|q_j\|^2}
=(v_1^\top q_j)^2.
\]
As such, $\epsilon>(v_1^\top q_j)^2\geq(v_1^\top q_0)^2$.
Overall, when $H_0$ holds, the power iteration detector rejects $H_0$ only if $q_0$ is initialized poorly, i.e., $(v_1^\top q_0)^2<\epsilon$, which occurs with probability $\pi_\epsilon$ (since $q_0$ has a rotation-invariant probability distribution).
For the false negative error rate, note that Proposition~\ref{prop.power iteration guarantee} gives that $H_1$ implies convergence $(v^\top q_j)^2\rightarrow1$ provided $q_0$ is not orthogonal to $v$, i.e., with probability $1$.

For (iii), we want $j$ such that $(v^\top q_j)^2>1-\epsilon$.
By Proposition~\ref{prop.power iteration guarantee}, it suffices to have
\[
\Big((v_1^\top q_0)^{-2}-1\Big)\bigg(\frac{\lambda_2}{\lambda_1}\bigg)^{2j}
<\epsilon.
\]
In the event that $(v_1^\top q_0)^2\geq\epsilon$ (which has probability $1-\pi_\epsilon$), it further suffices to have
\[
\epsilon^{-2}\bigg(\frac{\lambda_2}{\lambda_1}\bigg)^{2j}
<\epsilon.
\]
Taking logs and rearranging then gives the result.
\end{proof}

To estimate $\epsilon$ and $\pi_\epsilon$, first note that $q_0$ has a rotation-invariant probability distribution, and so linearity of expectation gives
\[
\mathbb{E}\big[(e_1^\top q_0)^2\big]
=\frac{1}{n}\sum_{i=1}^n\mathbb{E}\big[(e_i^\top q_0)^2\big]
=\frac{1}{n}\mathbb{E}\|q_0\|_2^2
=\frac{1}{n}.
\]
Thus, in order to make $\pi_\epsilon$ small, we should expect to have $\epsilon\ll1/n$.
The following lemma gives that such choices of $\epsilon$ suffice for $\pi_\epsilon$ to be small:

\begin{lemma}
If $\epsilon\geq n^{-1}e^{-2n}$, then $\pi_\epsilon\leq 3\sqrt{n\epsilon}$.
\end{lemma}

\begin{proof}
First, observe that $(e_1^\top q_0)^2$ is equal in distribution to $Z^2/Q$, where $Z$ has standard normal distribution and $Q$ has chi-squared distribution with $n$ degrees of freedom ($Z$ and $Q$ are independent).
The probability density function of $Z$ has a maximal value of $1/\sqrt{2\pi}$ at zero, and so
\[
\operatorname{Pr}\Big(Z^2<a\Big)
\leq\sqrt{\frac{2a}{\pi}}.
\]
Also, Lemma~1 in~\cite{laurent2000adaptive} gives
\[
\operatorname{Pr}\Big(Q\geq n+2\sqrt{nx}+2x\Big)\leq e^{-x}
\qquad
\forall x>0.
\]
Therefore, picking $a=5n\epsilon$ and $x=n$, the union bound gives
\[
\operatorname{Pr}\Big((e_1^\top q_0)^2<\epsilon\Big)
=\operatorname{Pr}\bigg(\frac{Z^2}{Q}<\epsilon\bigg)
\leq\operatorname{Pr}\Big(Z^2<5n\epsilon\Big)+\operatorname{Pr}\Big(Q>5n\Big)
\leq \sqrt{\frac{10n\epsilon}{\pi}}+e^{-n}
\leq 3\sqrt{n\epsilon}.\qedhere
\]
\end{proof}

Overall, if we take $\epsilon=n^{-(2c+1)}$ for $c>0$, then if $H_0$ is true, our detector will produce a false positive with probability $O(n^{-c})$.
On the other hand, if $H_1$ is true, then with probability $1-O(n^{-c})$, our detector will reject $H_0$ after $O_\delta(c\log n)$ power iterations, provided $|\lambda_2|\leq(1-\delta)|\lambda_1|$.

\subsection{Testing optimality with the power iteration detector}\label{sec.testing opt}

In this subsection, we leverage the power iteration detector to test $k$-means optimality.
Note that the sufficient condition \eqref{cert_condition} holds if and only if $v:=\frac{1}{\sqrt{N}}1$ is a leading eigenvector of the matrix
\begin{equation}
\label{eq.fast cert matrix}
A
:=\frac{z}{N}11^\top+P_{\Lambda^\perp}(B-M)P_{\Lambda^\perp}
=\frac{z}{N}11^\top+P_{\Lambda^\perp}(B-D)P_{\Lambda^\perp}.
\end{equation}
(The second equality follows from distributing the $P_{\Lambda^\perp}$'s and recalling the definition of $M$ in~\eqref{eq.definition of M}.)
As such, it suffices that $(A,v)$ satisfy $H_1$ in \eqref{eq.hypotheses}.
Overall, given a collection of points $\{x_i\}_{i=1}^N\subseteq\mathbb{R}^m$ and a proposed partition $A_1\sqcup\cdots\sqcup A_k=\{1,\ldots,N\}$, we can produce the corresponding matrix $A$ (defined above) and then run the power iteration detector of the previous subsection to test \eqref{cert_condition}.
In particular, a positive test with tolerance $\epsilon$ will yield $\geq1-\pi_\epsilon$ confidence that the proposed partition is optimal under the $k$-means objective.
Furthermore, as we detail below, the matrix--vector products computed in the power iteration detector have a computationally cheap implementation.

Given an $m\times n_a$ matrix $\Phi_a=[x_{a,1}\cdots x_{a,n_a}]$ for each $a\in\{1,\ldots,k\}$, we follow the following procedure to implement the corresponding function $x\mapsto Ax$ as defined in \eqref{eq.fast cert matrix}:
\begin{enumerate}
\item
Compute $\nu_a\in\mathbb{R}^{n_a}$ such that $(\nu_a)_i=\|x_{a,i}\|_2^2$ for every $a\in\{1,\ldots,k\}$ in $O(mN)$ operations.\\
Let $\nu\in\mathbb{R}^N$ denote the vector whose $a$th block is $\nu_a$.
\item
Define the function $(a,b,x)\mapsto D^{(a,b)}x$ such that $D^{(a,b)}=\nu_a1^\top-2\Phi_a^\top\Phi_b+1\nu_b^\top$.\\
Running this function costs $O(m(n_a+n_b))$ operations.
\item
Define the function $x\mapsto Dx$ such that $D=\nu1^\top-2\Phi^\top\Phi+1\nu^\top$, where $\Phi=[\Phi_1\cdots\Phi_k]$.\\
Running this function costs $O(mN)$ operations.
\item
Compute $\mu_a=\frac{1}{2}(\frac{1}{n_a^2}11^\top-\frac{2}{n_a}I)D^{(a,a)}1$ for every $a\in\{1,\ldots,k\}$ in $O(mN)$ operations.
\item
Define the function $(a,b,x)\mapsto M^{(a,b)}x$ such that $M^{(a,b)}=D^{(a,b)}+\mu_a1^\top+1\mu_b^\top$.\\
Running this function costs $O(m(n_a+n_b))$ operations.
\item
Compute $z=\min_{a\neq b}\frac{2n_a}{n_a+n_b}\min(M^{(a,b)}1)$ in $O(kmN)$ operations.
\item
Compute $u_{(a,b)}=M^{(a,b)}1-z\frac{n_a+n_b}{2n_a}1$ for every $a,b\in\{1,\ldots,k\}$, $a\neq b$ in $O(kmN)$ operations.
\item
Compute $\rho_{(a,b)}=u_{(a,b)}^\top1$ for every $a,b\in\{1,\ldots,k\}$, $a\neq b$ in $O(kN)$ operations.
\item
Define the function $x\mapsto Bx$ such that the $a$th block of the output is given by
\[
(Bx)_a
=\sum_{\substack{b=1\\b\neq a}}^k\frac{u_{(a,b)}u_{(b,a)}^\top x_b}{\rho_{(b,a)}}.
\]
Running this function costs $O(kmN)$ operations.
\item
Define the function $x\mapsto P_{\Lambda^\perp}x$ such that $P_{\Lambda^\perp}=I-\sum_{a=1}^k\frac{1}{n_a}1_a1_a^\top$.\\
Running this function costs $O(N)$ operations.
\item
Define the function $x\mapsto Ax$ such that $A=\frac{z}{N}11^\top+P_{\Lambda^\perp}(B-D)P_{\Lambda^\perp}$.\\
Running this function costs $O(kmN)$ operations.
\end{enumerate}
Overall, after $O(kmN)$ operations of preprocessing, one may compute the function $x\mapsto Ax$ for any given $x$ in $O(kmN)$ operations.
(Observe that this is the same complexity as each iteration of Lloyd's algorithm, and as we illustrate in Figure~\ref{figure.2}, the runtimes are comparable.)

At this point, we take a short aside to illustrate the utility of the power iteration detector beyond $k$-means clustering.
The original problem for which a PCC algorithm was developed was community recovery under the \textbf{stochastic block model}~\cite{bandeira2015note}.
For this random graph, there are two communities of vertices, each of size $n/2$, and edges are drawn independently at random with probability $p$ if the pair of vertices belong to the same community, and with probability $q<p$ if they come from different communities.
Given the random edges, the maximum likelihood estimator for the communities is given by the vertex partition of two sets of size $n/2$ with the minimum cut.
Given a partition of the vertices, let $X$ denote the corresponding $n\times n$ matrix of $\pm1$s such that $X_{ij}=1$ precisely when $i$ and $j$ belong to the same community.
Given the adjacency matrix $A$ of the random graph, one may express the cut of a partition $X$ in terms of $\operatorname{Tr}(AX)$.
Furthermore, $X$ satisfies the convex constraints $X_{ii}=1$ and $X\succeq0$, and so one may relax to these constraints to obtain a semidefinite program and hope that the relaxation is typically tight over a large region of $(p,q)$.
Amazingly, this relaxation is typically tight precisely over the region of $(p,q)$ for which community recovery is information-theoretically possible~\cite{abbe2014exact}.

Given $A$, put $B:=2A-11^\top+I$, and given a vector $x\in\mathbb{R}^n$, define the corresponding $n\times n$ diagonal matrix $D_x$ by $(D_x)_{ii}:=x_i\sum_{j=1}^n B_{ij}x_j$.
In~\cite{bandeira2015note}, Bandeira observes that, given a partition matrix $X$ by some means (such as the fast algorithm provided in~\cite{abbe2015community}), then $X=xx^\top$ is SDP-optimal if both $x^\top1=0$ and the second smallest eigenvalue of $D_x-B$ is strictly positive, meaning the partition gives the maximum likelihood estimator for the communities.
However, as Bandeira notes, the computational bottleneck here is estimating the second smallest eigenvalue of $D_x-B$, and he suggests that a randomized power method--like algorithm might suffice, but leaves the investigation for future research.

Here, we show how the power iteration detector fills this void in the theory.
First, we note that in the interesting regime of $(p,q)$, the number of nonzero entries in $A$ is $O(n\log n)$ with high probability~\cite{abbe2014exact}.
As such, the function $x\mapsto Bx$ can exploit this sparsity to take only $O(n\log n)$ operations.
This in turn allows for the computation of the diagonal of $D_x$ to cost $O(n\log n)$ operations.
Next, note that 
\begin{align*}
\|D_x-B\|_{2\rightarrow2}
&\leq\|D_x\|_{2\rightarrow2}+\|2A-11^\top\|_{2\rightarrow2}+\|I\|_{2\rightarrow2}\\
&\leq\|D_x\|_{2\rightarrow2}+\|2A-11^\top\|_F+1
=\max_i|(D_x)_{ii}|+n+1
=:\lambda,
\end{align*}
and that $\lambda$ can be computed in $O(n)$ operations after computing the diagonal of $D_x$.
Also, it takes $O(n)$ operations to verify $x^\top1=0$.
Assuming $x^\top1=0$, then the second smallest eigenvalue of $D_x-B$ is strictly positive if and only if $x$ spans the unique leading eigenspace of $\lambda I-D_x+B$.
Thus, one may test this condition using the power iteration detector, and furthermore, each iteration will take only $O(n\log n)$ operations, thanks to the sparsity of $A$.

\section{A fast $k$-means solver for two clusters}\label{sec:fast_solver}

The previous section illustrated how to quickly test whether a proposed solution to the $k$-means problem is optimal.
In particular, this test will be successful with high probability if the data follows the stochastic ball model with $\Delta>2+k^2/m$.
It remains to find a fast $k$-means solver which also performs in this regime.

In doing so, we maintain the philosophy that our algorithm should not ``see'' the stochastic ball model.
Indeed, we view the stochastic ball model as a method of evaluating clustering algorithms rather than a realistic data model.
For example, Lloyd's algorithm can be viewed as an alternating minimization of the lifted objective function:
\[
f(A_1,\ldots,A_k,c_1,\ldots,c_k)
:=\sum_{t=1}^k\sum_{i\in A_t}\|x_i-c_t\|^2,
\qquad
A_1\sqcup\cdots\sqcup A_k=\{1,\ldots,N\},~c_1,\ldots,c_k\in\mathbb{R}^m,
\]
and since this function is minimized at the $k$-means optimizer (regardless of how the data is distributed), such an algorithm is acceptable.
On the other hand, one might consider matching the stochastic ball model to the data by maximizing the following function:
\[
g(c_1,\ldots,c_k)
:=\sum_{i=1}^N\sum_{t=1}^k p_\mathcal{D}(x_i-c_t),
\qquad
c_1,\ldots,c_k\in\mathbb{R}^m,
\]
where $p_\mathcal{D}(\cdot)$ denotes the density function of $\mathcal{D}$, which is supported on the unit ball centered at the origin.
One could certainly devise a fast greedy method such as matching pursuit~\cite{mallat1993matching} to optimize this objective function (especially if $p_\mathcal{D}$ is smooth), but doing so violates our philosophy.

In~\cite{peng2007approximating}, Peng and Wei showed that $k$-means is equivalent to the following program:
\begin{alignat}{2}
\label{eq.kmeans_program}
& \text{minimize}  &       & \operatorname{Tr}(DX) \\
\nonumber
& \text{subject to}& \quad & 
\begin{aligned}[t]
X^\top&=X,~
X^2=X,~
\operatorname{Tr}(X)=k,~
X1=1,~
X\geq0
\end{aligned}
\end{alignat}
One may quickly observe that the SDP~\eqref{eq.kmeansSDP} we analyzed in Section~\ref{sec:sdp} is a relaxation of this program.
In this section, we follow Peng and Wei~\cite{peng2007approximating} by considering another relaxation of \eqref{eq.kmeans_program}, obtained by discarding the $X\geq0$ constraint (this is known as the \textbf{spectral clustering} relaxation~\cite{dhillon2004kernel,dhillon2007weighted}).
We first denote the $m\times N$ matrix $\Phi=[x_1\cdots x_N]$.
Without loss of generality, the data set is centered at the origin so that $\Phi1=0$.
Letting $\nu$ denote the $N\times 1$ vector with $\nu_i=\|x_i\|_2^2$, then
\[
D_{ij}
=\|x_i-x_j\|_2^2
=\|x_i\|_2^2-2x_i^\top x_j+\|x_j\|_2^2
=(\nu1^\top-2\Phi^\top\Phi+1\nu^\top)_{ij}.
\]
As such, $D=\nu1^\top-2\Phi^\top\Phi+1\nu^\top$, and so the constraints $X=X^\top$ and $X1=1$ together imply an alternative expression for the objective function:
\begin{align*}
\operatorname{Tr}(DX)
&=\operatorname{Tr}(\nu1^\top X-2\Phi^\top\Phi X+1\nu^\top X)\\
&=\operatorname{Tr}(\nu1^\top X^\top)-2\operatorname{Tr}(\Phi^\top\Phi X)+\operatorname{Tr}(X1\nu^\top)\\
&=2\nu^\top 1-2\operatorname{Tr}(\Phi^\top\Phi X).
\end{align*}
We conclude that minimizing $\operatorname{Tr}(DX)$ is equivalent to maximizing $\operatorname{Tr}(\Phi^\top\Phi X)$.

Next, we observe that the feasible $X$ in our relaxation are precisely the rank-$k$ $N\times N$ orthogonal projection matrices satisfying $X1=1$.
This in turn is equivalent to $X$ having the form $X=\frac{1}{N}11^\top+Y$, where $Y$ is a rank-$(k-1)$ $N\times N$ orthogonal projection matrix satisfying $Y1=0$.
Discarding the $Y1=0$ constraint produces the following relaxation of \eqref{eq.kmeans_program}:
\begin{alignat}{2}
\label{eq.kmeans_spectral}
& \text{maximize}  &       & \operatorname{Tr}(\Phi^\top\Phi Y) \\
\nonumber
& \text{subject to}& \quad & 
\begin{aligned}[t]
Y^\top&=Y,~
Y^2=Y,~
\operatorname{Tr}(Y)=k-1
\end{aligned}
\end{alignat}
For general values of $k$, this program amounts to finding $k-1$ principal components of the data.
Recalling our initial clustering goal, after finding the optimal $Y$, it remains to take $X=\frac{1}{N}11^\top+Y$ and then round to a nearby member of the feasibility region in \eqref{eq.kmeans_program}.
In~\cite{peng2007approximating}, Peng and Wei focus on the $k=2$ case; they reduce the rounding step to a $2$-means problem on the real line, and they establish an approximation ratio of $2$ for this relax-and-round procedure.
Here, we are concerned with exact recovery under the stochastic ball model, and as such, we slightly modify the rounding step.

When $k=2$, the solution to \eqref{eq.kmeans_spectral} has the form $Y=yy^\top$, where $y$ is a leading unit eigenvector of $\Phi^\top\Phi$.
Our task is to find a matrix of the form $\frac{1}{|A|}1_A1_A^\top+\frac{1}{|B|}1_B1_B^\top$ with $A\sqcup B=\{1,\ldots,N\}$ that is close to $\frac{1}{N}11^\top+yy^\top$.
To this end, it seems natural to consider
\[
A_\theta:=\{i:y_i<\theta\},
\qquad
B_\theta:=A_\theta^c
\]
for some threshold $\theta$.
Since the data is centered ($\Phi1=0$), one may be inclined to take $\theta=0$, but this will be a poor choice if the true clusters have significantly different numbers of points.
Instead, we select the $\theta$ which minimizes the $k$-means objective of $(A_\theta,B_\theta)$.
Since we only need to consider $N-1$ choices of $\theta$, this is plausibly tractable, although computing the $k$-means objective once costs $O(mN)$ operations, and so some care is necessary to keep the algorithm fast.

We will show how to find the optimal $(A_\theta,B_\theta)$ in $O((m+\log N)N)$ operations using a simple dynamic program.
Order the indices so that $y_1\leq\cdots\leq y_N$.
Then the function to minimize is
\[
f(i):=\frac{1}{i}\underbrace{\sum_{j=1}^i\sum_{j'=1}^i\|x_j-x_{j'}\|_2^2}_{v_i}+\frac{1}{N-i}\underbrace{\sum_{j=i+1}^N\sum_{j'=i+1}^N\|x_j-x_{j'}\|_2^2}_{v^c_i}.
\]
Expanding the square and distributing sums gives
\[
v_{i+1}
=v_i+2\sum_{j=1}^i\|x_j\|_2^2-4x_{i+1}^\top\sum_{j=1}^ix_j+2i\|x_{i+1}\|_2^2,
\]
and the $v^c_i$'s satisfy a similar recursion rule.
As such, one may iteratively compute the $v_i$'s and $v^c_i$'s before computing the $f(i)$'s and then minimizing.
Overall, the following procedure finds the optimal $(A_\theta,B_\theta)$ in $O((m+\log N)N)$ operations:
\begin{enumerate}
\item
Sort the entries $y_1\leq\cdots\leq y_N$ in $O(N\log N)$ operations.
\item
Iteratively compute
\[
s_1(i):=\sum_{j=1}^ix_j,
\quad
s_1^c(i):=\sum_{j=i+1}^Nx_j,
\quad
s_2(i):=\sum_{j=1}^i\|x_j\|_2^2,
\quad
s_2^c(i):=\sum_{j=i+1}^N\|x_j\|_2^2
\]
for every $i\in\{1,\ldots,N-1\}$ in $O(mN)$ operations.
\item
Compute $v_1=0$ and $v_{i+1}=v_i+2s_2(i)-4x_{i+1}^\top s_1(i)+2i\|x_{i+1}\|_2^2$ for every $i\in\{1,\ldots,N-2\}$ in $O(mN)$ operations.
\item
Compute $v^c_{N-1}=0$ and $v^c_{i-1}=v^c_i+2s_2^c(i)-4x_i^\top s_1^c(i)+2(N-i)\|x_i\|_2^2$ for every $i\in\{N-1,\ldots,2\}$ in $O(mN)$ operations.
\item
Compute $f(i)=v_i/i+v^c_i/(N-i)$ for every $i\in\{1,\ldots,N-1\}$ in $O(N)$ operations.
\item
Find $i$ that minimizes $f(i)$ and output $\{1,\ldots,i\}$ and $\{i+1,\ldots,N\}$ in $O(N)$ operations.
\end{enumerate}
Note that in the special case where $m=1$, the above method exactly solves the $k$-means problem when $k=2$ in only $O(N\log N)$ operations, recovering the rounding step of Peng and Wei~\cite{peng2007approximating}.
For comparison, \cite{wang2011ckmeans} leverages more sophisticated dynamic programming for the $m=1$ case, but $k$ is arbitrary and the algorithm costs $O(kN^2)$ operations.

\begin{algorithm}[t]
\caption{Spectral $k$-means clustering (for two clusters)}
\label{alg.euclidean spectral clustering}
\SetAlgoLined
\KwIn{$m\times N$ matrix $\Phi=[x_1\cdots x_N]$ of points to be clustered}
\KwOut{Clusters $A\sqcup B=\{1,\ldots,N\}$}
Subtract centroid $\frac{1}{N}\sum_{i=1}^Nx_i$ from each column of $\Phi$ to produce $\Phi_0$\\
Compute leading eigenvector $y$ of $\Phi_0^\top\Phi_0$\\
Find $\theta$ that minimizes the $k$-means objective of $(\{i:y_i<\theta\},\{i:y_i\geq\theta\})$\\
$(A,B)\leftarrow(\{i:y_i<\theta\},\{i:y_i\geq\theta\})$
\end{algorithm}

See Algorithm~\ref{alg.euclidean spectral clustering} for a summary of our relax-and-round procedure.
As a spectral method, this algorithm enjoys quasilinear computational complexity; see Figure~\ref{figure.2} for an illustration.
In particular, when computing the leading eigenvector of $\Phi_0^\top\Phi_0$, each matrix--vector multiply in the power method costs only $O(mN)$ operations.
Furthermore, as the following result guarantees, this algorithm performs well under the stochastic ball model:

\begin{theorem}
\label{theorem.spectral clustering}
Let $\Delta^\star=\Delta^\star(\mathcal{D},k)$ denote the smallest value for which $\Delta>\Delta^\star$ implies that minimizing the $k$-means objective recovers planted clusters under the $(\mathcal{D},\gamma,n)$-stochastic ball model with probability $1-e^{-\Omega_{\mathcal{D},\gamma}(n)}$.
When $k=2$, spectral $k$-means clustering (Algorithm~\ref{alg.euclidean spectral clustering}) recovers planted clusters under the stochastic ball model with probability $1-e^{-\Omega_{\mathcal{D},\gamma}(n)}$ provided $\Delta>\Delta^\star$.
\end{theorem}

See Appendix~\ref{sec:appendix_theorem_spectral_clustering} for the proof.
The main idea is that the leading eigenvector of $\Phi_0\Phi_0^\top$ is biased towards the difference between the ball centers, and as the following lemma establishes, this bias encourages spectral $k$-means clustering to separate the planted clusters:

\begin{lemma}
\label{lemma.sufficient eigenvector}
Take two clusters contained in unit balls centered at $\gamma$ and $-\gamma$ with $\|\gamma\|_2>1$.
If minimizing the $k$-means objective recovers these clusters, then spectral $k$-means clustering (Algorithm~\ref{alg.euclidean spectral clustering}) also recovers them, provided the leading eigenvector $z$ of $\Phi_0\Phi_0^\top$ satisfies $|\gamma^\top z|>\|z\|_2$.
\end{lemma}

\begin{proof}
Write $\Phi_0=\Phi-\mu1^\top$, put $\theta:=-\mu^\top z$, and observe that $y=\Phi_0^\top z$ is a leading eigenvector of $\Phi_0^\top\Phi_0$.
Then
\begin{equation}
\label{eq.separating}
y_i=(x_i-\mu)^\top z=x_i^\top z+\theta
\end{equation}
for every $i$.
Next, if $|\gamma^\top z|>\|z\|_2$, then a simple trigonometric argument gives that the balls (and therefore the planted clusters) are separated by the hyperplane orthogonal to $z$.
Combined with \eqref{eq.separating}, we then have that the clusters can be identified according to whether $y_i<\theta$ or $y_i>\theta$.
It therefore suffices to minimize the $k$-means objective subject to partitions of this form (for arbitrary thresholds $\theta$), as so spectral $k$-means clustering succeeds. 
\end{proof}

\section{Discussion}\label{sec:discussion}

This paper discussed various facets of probably certifiably correct algorithms for $k$-means clustering.
There are still many questions that have yet to be answered:
\begin{itemize}
\item
Let $\Delta^\star(\mathcal{D},k)$ denote the smallest value for which $\Delta>\Delta^\star$ implies that minimizing the $k$-means objective recovers planted clusters under the $(\mathcal{D},\gamma,n)$-stochastic ball model with probability $1-e^{-\Omega_{\mathcal{D},\gamma}(n)}$.
What is $\Delta^\star$?
It was conjectured in \cite{relax} that $\Delta^\star=2$, but as we demonstrated in Subsection~\ref{sec:SBM_SDP}, this is not the case.
\item
Let $\Delta_\mathrm{SDP}^\star(\mathcal{D},k)$ denote the smallest value for which $\Delta>\Delta_\mathrm{SDP}^\star$ implies that solving the $k$-means SDP recovers planted clusters under the $(\mathcal{D},\gamma,n)$-stochastic ball model with probability $1-e^{-\Omega_{\mathcal{D},\gamma}(n)}$.
What is $\Delta_\mathrm{SDP}^\star$?
Considering Subsection~\ref{sec:SBM_SDP} and Figure~\ref{figure.1}(center), we suspect the SDP exhibits a performance gap: $\Delta_\mathrm{SDP}^\star>\Delta^\star$.
\item
Is there a single dual certificate for the $k$-means SDP that typically certifies planted clusters under the stochastic ball model whenever $\Delta>\Delta_\mathrm{SDP}^\star$?
Does this certification have a quasilinear-time implementation similar to Subsection~\ref{sec.testing opt}?
\item
Is there a quasilinear-time $k$-means solver that typically solves $k$-means under the stochastic ball model whenever $\Delta>\Delta^\star$?
In particular, is there a quasilinear-time initialization of Lloyd's algorithm that meets this specification?
Following the philosophy of Section~\ref{sec:fast_solver}, such algorithms should be designed so as to not ``see'' the stochastic ball model.
\end{itemize}

\section*{Acknowledgments}

The authors thank the anonymous referees, whose suggestions significantly improved this paper's presentation and literature review.
The authors also thank Afonso S.\ Bandeira and Nicolas Boumal for interesting discussions and valuable comments on an earlier version of this manuscript.
DGM was supported by an AFOSR Young Investigator Research Program award, NSF Grant No.\ DMS-1321779, and AFOSR Grant No.\ F4FGA05076J002.
SV was supported by Rachel Ward's NSF CAREER award and AFOSR Young Investigator Research Program award.
The views expressed in this article are those of the authors and do not reflect the official policy or position
of the United States Air Force, Department of Defense, or the U.S.\ Government.

\bibliographystyle{abbrv}
\bibliography{PCC_kmeans}


\appendix

\section{Proof of Corollary \ref{cor.dual certificate}} \label{sec:appendix_corollary}

It suffices to have
\begin{equation}
\label{eq.triangle for new condition}
\|P_{\Lambda^\perp}MP_{\Lambda^\perp}\|_{2\rightarrow2}+\|P_{\Lambda^\perp}BP_{\Lambda^\perp}\|_{2\rightarrow2}
\leq z.
\end{equation}
We will bound the terms in \eqref{eq.triangle for new condition} separately and then combine the bounds to derive a sufficient condition for Theorem~\ref{thm.dual certificate}.
To bound the first term in \eqref{eq.triangle for new condition}, let $\nu$ be the $N\times 1$ vector whose $(a,i)$th entry is $\|x_{a,i}\|_2^2$, and let $\Phi$ be the $m\times N$ matrix whose $(a,i)$th column is $x_{a,i}$.
Then
\[
D_{(a,i),(b,j)}
=\|x_{a,i}-x_{b,j}\|_2^2
=\|x_{a,i}\|_2^2-2x_{a,i}^\top x_{b,j}+\|x_{b,j}\|_2^2
=(\nu1^\top-2\Phi^\top\Phi+1\nu^\top)_{(a,i),(b,j)},
\]
meaning $D=\nu1^\top-2\Phi^\top\Phi+1\nu^\top$.
With this, we appeal to the blockwise definition of $M$ \eqref{eq.definition of M}:
\begin{align*}
\|P_{\Lambda^\perp}MP_{\Lambda^\perp}\|_{2\rightarrow2}
=\|P_{\Lambda^\perp}DP_{\Lambda^\perp}\|_{2\rightarrow2}
&=\|P_{\Lambda^\perp}(\nu1^\top-2\Phi^\top\Phi+1\nu^\top)P_{\Lambda^\perp}\|_{2\rightarrow2}\\
&=2\|P_{\Lambda^\perp}\Phi^\top\Phi P_{\Lambda^\perp}\|_{2\rightarrow2}
=2\|\Phi P_{\Lambda^\perp}\|_{2\rightarrow2}^2
=2\|\Psi\|_{2\rightarrow2}^2.
\end{align*}
For the second term in \eqref{eq.triangle for new condition}, we first write the decomposition
\[
B=\sum_{a=1}^k\sum_{b=a+1}^k\Big(H_{(a,b)}(B^{(a,b)})+H_{(b,a)}(B^{(b,a)})\Big),
\]
where $H_{(a,b)}\colon\mathbb{R}^{n_a\times n_b}\rightarrow\mathbb{R}^{N\times N}$ produces a matrix whose $(a,b)$th block is the input matrix, and is otherwise zero.
Then
\begin{align*}
P_{\Lambda^\perp}BP_{\Lambda^\perp}
&=\sum_{a=1}^k\sum_{b=a+1}^kP_{\Lambda^\perp}\Big(H_{(a,b)}(B^{(a,b)})+H_{(b,a)}(B^{(b,a)})\Big)P_{\Lambda^\perp}\\
&=\sum_{a=1}^k\sum_{b=a+1}^k\Big(H_{(a,b)}(P_{1^\perp}B^{(a,b)}P_{1^\perp})+H_{(b,a)}(P_{1^\perp}B^{(b,a)}P_{1^\perp})\Big),
\end{align*}
and so the triangle inequality gives
\begin{align*}
\|P_{\Lambda^\perp}BP_{\Lambda^\perp}\|_{2\rightarrow2}
&\leq\sum_{a=1}^k\sum_{b=a+1}^k\|H_{(a,b)}(P_{1^\perp}B^{(a,b)}P_{1^\perp})+H_{(b,a)}(P_{1^\perp}B^{(b,a)}P_{1^\perp})\|_{2\rightarrow2}\\
&=\sum_{a=1}^k\sum_{b=a+1}^k\|P_{1^\perp}B^{(a,b)}P_{1^\perp}\|_{2\rightarrow2},
\end{align*}
where the last equality can be verified by considering the spectrum of the square:
\begin{align*}
&\Big(H_{(a,b)}(P_{1^\perp}B^{(a,b)}P_{1^\perp})+H_{(b,a)}(P_{1^\perp}B^{(b,a)}P_{1^\perp})\Big)^2\\
&\qquad=H_{(a,a)}\Big((P_{1^\perp}B^{(a,b)}P_{1^\perp})(P_{1^\perp}B^{(a,b)}P_{1^\perp})^\top\Big)+H_{(b,b)}\Big((P_{1^\perp}B^{(a,b)}P_{1^\perp})^\top(P_{1^\perp}B^{(a,b)}P_{1^\perp})\Big).
\end{align*}
At this point, we use the definition of $B$ \eqref{eq.how to construct B} to get
\[
\|P_{1^\perp}B^{(a,b)}P_{1^\perp}\|_{2\rightarrow2}
=\frac{\|P_{1^\perp}u_{(a,b)}\|_2\|P_{1^\perp}u_{(b,a)}\|_2}{\rho_{(a,b)}}.
\]
Recalling the definition of $u_{(a,b)}$ \eqref{eq.how to construct B} and combining these estimates then produces the result.

\section{Proof Theorem \ref{main_theorem}} \label{sec:model} \label{sec:appendix_theorem}

In this section, we apply the certificate from Corollary \ref{cor.dual certificate} to the  $(\mathcal{D},\gamma,n)$-stochastic ball model (see Definition \ref{stochastic_balls}) to prove our main result. We will prove Theorem \ref{main_theorem} with the help of several lemmas.

\begin{lemma}
Denote
\[
c_a:=\frac{1}{n}\sum_{i=1}^nx_{a,i},
\qquad
\Delta_{ab}:=\|\gamma_a-\gamma_b\|_2,
\qquad
O_{ab}:=\frac{\gamma_a+\gamma_b}{2}.
\]
Then the $(\mathcal{D},\gamma,n)$-stochastic ball model satisfies the following estimates:
\begin{align}
\label{eq.empirical center}
\|c_a-\gamma_a\|_2&<\epsilon\qquad\mbox{w.p.}\qquad1-e^{-\Omega_{m,\epsilon}(n)}\\
\label{eq.empirical average radius}
\bigg|\frac{1}{n}\sum_{i=1}^n\|r_{a,i}\|_2^2-\mathbb{E}\|r\|_2^2\bigg|&<\epsilon\qquad\mbox{w.p.}\qquad1-e^{-\Omega_\epsilon(n)}\\
\label{eq.empirical average distance from midpoint}
\bigg|\frac{1}{n}\sum_{i=1}^n\|x_{a,i}-O_{ab}\|_2^2-\mathbb{E}\|r+\gamma_a-O_{ab}\|_2^2\bigg|&<\epsilon\qquad\mbox{w.p.}\qquad1-e^{-\Omega_{\Delta_{ab},\epsilon}(n)}
\end{align}
\end{lemma}

\begin{proof}
Since $\mathbb{E}r=0$ and $\|r\|_2^2\leq1$ almost surely, one may lift
\[
X_{a,i}:=\left[\begin{array}{cc}0&r_{a,i}^\top\\r_{a,i}&0\end{array}\right]
\]
and apply the Matrix Hoeffding inequality~\cite{tropp2012user} to conclude that
\[
\operatorname{Pr}\bigg(\bigg\|\sum_{i=1}^nr_{a,i}\bigg\|_2\geq t\bigg)\leq me^{-t^2/8n}.
\]
Taking $t:=\epsilon n$ then gives \eqref{eq.empirical center}.
For \eqref{eq.empirical average radius} and \eqref{eq.empirical average distance from midpoint}, notice that the random variables in each sum are iid and confined to an interval almost surely, and so the result follows from Hoeffding's inequality.
\end{proof}

\begin{lemma}
\label{lemma.difference of distances}
Under the $(\mathcal{D},\gamma,n)$-stochastic ball model, we have $D^{(a,b)}1-D^{(a,a)}1=4np+q$, where
\begin{align*}
	p_i	&:=r_{a,i}^\top(\gamma_a-O_{ab})+\frac{\Delta_{ab}^2}{4}\\
	q_i	&:=2n(x_{a,i}-O_{ab})^\top\bigg((c_a-c_b)-(\gamma_a-\gamma_b)\bigg)+\bigg(\sum_{j=1}^n\|x_{b,j}-O_{ab}\|_2^2-\sum_{j=1}^n\|x_{a,j}-O_{ab}\|_2^2\bigg)
\end{align*}
and $|q_i|\leq(6+2\Delta_{ab})n\epsilon$ with probability $1-e^{-\Omega_{m,\Delta_{ab},\epsilon}(n)}$.
\end{lemma}

\begin{proof}
Add and subtract $O_{ab}$ and then expand the squares to get
\begin{align*}
e_i^\top(D^{(a,b)}1-D^{(a,a)}1)
&=\sum_{j=1}^n\|x_{a,i}-x_{b,j}\|_2^2-\sum_{j=1}^n\|x_{a,i}-x_{a,j}\|_2^2\\
&=n\bigg(-2(x_{a,i}-O_{ab})^\top(c_b-O_{ab})+\frac{1}{n}\sum_{j=1}^n\|x_{b,j}-O_{ab}\|_2^2\bigg)\\
&\qquad-n\bigg(-2(x_{a,i}-O_{ab})^\top(c_a-O_{ab})+\frac{1}{n}\sum_{j=1}^n\|x_{a,j}-O_{ab}\|_2^2\bigg)\\
&=2n(x_{a,i}-O_{ab})^\top(c_a-c_b)+\bigg(\sum_{j=1}^n\|x_{b,j}-O_{ab}\|_2^2-\sum_{j=1}^n\|x_{a,j}-O_{ab}\|_2^2\bigg).
\end{align*}
Add and subtract $\gamma_a-\gamma_b$ to $c_a-c_b$ and distribute over the resulting sum to obtain
\begin{align*}
e_i^\top(D^{(a,b)}1-D^{(a,a)}1)
&=2n(x_{a,i}-O_{ab})^\top(\gamma_a-\gamma_b)+q\\
&=4n\Big(r_{a,i}+(\gamma_a-O_{ab})\Big)^\top(\gamma_a-O_{ab})+q.
\end{align*}
Distributing and identifying $\|\gamma_a-O_{ab}\|_2^2=\Delta_{ab}^2/4$ explains the definition of $p$.
To show $|q_i|\leq(6+2\Delta_{ab})n\epsilon$, apply triangle and Cauchy--Schwarz to obtain
\begin{align*}
|q_i|
	&\leq \bigg|2n(x_{a,i}-O_{ab})^\top\bigg((c_a-c_b)-(\gamma_a-\gamma_b)\bigg)\bigg|+\bigg|\sum_{j=1}^n\|x_{b,j}-O_{ab}\|_2^2-\sum_{j=1}^n\|x_{a,j}-O_{ab}\|_2^2\bigg|\\
	&\leq 2n \bigg(\|r_{a,i}\|_2+\|\gamma_a-O_{a,b}\|_2\bigg)\bigg(\|c_a-\gamma_a\|_2+\|c_b-\gamma_b\|_2\bigg)+\bigg|\sum_{j=1}^n\|x_{b,j}-O_{ab}\|_2^2-\sum_{j=1}^n\|x_{a,j}-O_{ab}\|_2^2\bigg|\\
	&\leq 2n\bigg(1+\frac{\Delta_{ab}}{2}\bigg)\bigg(\|c_a-\gamma_a\|_2+\|c_b-\gamma_b\|_2\bigg)+\bigg|\sum_{j=1}^n\|x_{b,j}-O_{ab}\|_2^2-\sum_{j=1}^n\|x_{a,j}-O_{ab}\|_2^2\bigg|.
\end{align*}
To finish the argument, apply \eqref{eq.empirical center} to the first term while adding and subtracting
\[
\mathbb{E}\|r+\gamma_a-O_{ab}\|_2^2=\mathbb{E}\|r+\gamma_b-O_{ab}\|_2^2,
\]
from the second and apply \eqref{eq.empirical average distance from midpoint}.
\end{proof}

\begin{lemma}
\label{lemma.within cluster distances}
Under the $(\mathcal{D},\gamma,n)$-stochastic ball model, we have
\[
\bigg|\frac{1}{n}1^\top D^{(a,a)}1-2n\mathbb{E}\|r\|_2^2\bigg|
\leq 4n\epsilon
\qquad
\mbox{w.p.}
\qquad
1-e^{-\Omega_{\Delta_{ab},\epsilon}(n)}.
\]
\end{lemma}

\begin{proof}
Add and subtract $\gamma_a$ and expand the square to get
\[
\frac{1}{n}e_i^\top D^{(a,a)}1
=\frac{1}{n}\sum_{j=1}^n\|x_{a,i}-x_{a,j}\|_2^2
=\|r_{a,i}\|_2^2-2r_{a,i}^\top(c_a-\gamma_a)+\frac{1}{n}\sum_{j=1}^n\|r_{a,j}\|_2^2.
\]
The triangle and Cauchy--Schwarz inequalities then give
\begin{align*}
&\bigg|\frac{1}{n}1^\top D^{(a,a)}1-2n\mathbb{E}\|r\|_2^2\bigg|\\
&\qquad=\bigg|\sum_{i=1}^n\bigg(\|r_{a,i}\|_2^2-2r_{a,i}^\top(c_a-\gamma_a)+\frac{1}{n}\sum_{j=1}^n\|r_{a,j}\|_2^2\bigg)-2n\mathbb{E}\|r\|_2^2\bigg|\\
&\qquad\leq n\bigg|\frac{1}{n}\sum_{i=1}^n\|r_{a,i}\|_2^2-\mathbb{E}\|r\|_2^2\bigg|+2\sum_{i=1}^n|r_{a,i}^\top(c_a-\gamma_a)|+n\bigg|\frac{1}{n}\sum_{j=1}^n\|r_{a,j}\|_2^2-\mathbb{E}\|r\|_2^2\bigg|\\
&\qquad\leq n\bigg|\frac{1}{n}\sum_{i=1}^n\|r_{a,i}\|_2^2-\mathbb{E}\|r\|_2^2\bigg|+2\sum_{i=1}^n\|c_a-\gamma_a\|_2+n\bigg|\frac{1}{n}\sum_{j=1}^n\|r_{a,j}\|_2^2-\mathbb{E}\|r\|_2^2\bigg|\\
&\qquad\leq 4n\epsilon,
\end{align*}
where the last step occurs with probability $1-e^{-\Omega_{\Delta_{ab},\epsilon}(n)}$ by a union bound over \eqref{eq.empirical average radius} and \eqref{eq.empirical center}.
\end{proof}

\begin{lemma}
\label{lemma.difference of average distances}
Under the $(\mathcal{D},\gamma,n)$-stochastic ball model, we have
\[
1^\top D^{(a,b)}1-1^\top D^{(a,a)}1
\geq n^2\Delta_{ab}^2-(6+4\Delta_{ab})n^2\epsilon
\qquad
\mbox{w.p.}
\qquad
1-e^{-\Omega_{m,\Delta_{ab},\epsilon}(n)}.
\]
\end{lemma}

\begin{proof}
Lemma~\ref{lemma.difference of distances} gives
\begin{align*}
1^\top D^{(a,b)}1-1^\top D^{(a,a)}1
&=1^\top(4np+q)\\
&\geq 4n\sum_{i=1}^n\bigg(r_{a,i}^\top(\gamma_a-O_{ab})+\frac{\Delta_{ab}^2}{4}\bigg)-(6+2\Delta_{ab})n^2\epsilon\\
&\geq 4n\bigg(n(c_a-\gamma_a)^\top(\gamma_a-O_{ab})+\frac{n\Delta_{ab}^2}{4}\bigg)-(6+2\Delta_{ab})n^2\epsilon.
\end{align*}
Cauchy--Schwarz along with \eqref{eq.empirical center} then gives the result.
\end{proof}

\begin{lemma}
\label{lemma.bound on rhs}
Under the $(\mathcal{D},\gamma,n)$-stochastic ball model, there exists $C=C(\gamma)$ such that
\[
\min_{\substack{a,b\in\{1,\ldots,k\}\\a\neq b}}\min(M^{(a,b)}1)
\geq n\Delta(\Delta-2)+Cn\epsilon
\qquad
\mbox{w.p.}
\qquad
1-e^{-\Omega_{m,\gamma,\epsilon}(n)},
\]
where $\displaystyle{\Delta:=\min_{\substack{a,b\in\{1,\ldots,k\}\\a\neq b}}\Delta_{ab}}$.
\end{lemma}

\begin{proof}
Fix $a$ and $b$.
Then by Lemma~\ref{lemma.difference of distances}, the following holds with probability $1-e^{-\Omega_{m,\Delta_{ab},\epsilon}(n)}$:
\begin{align*}
\min\Big(D^{(a,b)}1-D^{(a,a)}1\Big)
&\geq 4n\min_{i\in\{1,\ldots,n\}}\bigg(r_{a,i}^\top(\gamma_a-O_{ab})+\frac{\Delta_{ab}^2}{4}\bigg)-(6+2\Delta_{ab})n\epsilon\\
&\geq n\Delta_{ab}^2-2n\Delta_{ab}-(6+2\Delta_{ab})n\epsilon,
\end{align*}
where the last step is by Cauchy--Schwarz.
Taking a union bound with Lemma~\ref{lemma.within cluster distances} then gives
\begin{align*}
&\min(M^{(a,b)}1)\\
&\qquad=\min\Big(D^{(a,b)}1-D^{(a,a)}1\Big)+\frac{1}{2}\bigg(\frac{1}{n}1^\top D^{(a,a)}1-\frac{1}{n}1^\top D^{(b,b)}1\bigg)\\
&\qquad\geq\min\Big(D^{(a,b)}1-D^{(a,a)}1\Big)
-\frac{1}{2}\bigg(\bigg|\frac{1}{n}1^\top D^{(a,a)}1-2n\mathbb{E}\|r\|_2^2\bigg|+\bigg|\frac{1}{n}1^\top D^{(b,b)}1-2n\mathbb{E}\|r\|_2^2\bigg|\bigg)\\
&\qquad\geq n\Delta_{ab}(\Delta_{ab}-2)-(10+2\Delta_{ab})n\epsilon
\end{align*}
with probability $1-e^{-\Omega_{\Delta_{ab},\epsilon}(n)}$.
The result then follows from a union bound over $a$ and $b$.
\end{proof}

\begin{lemma}
\label{lemma.bound on numerator}
Suppose $\epsilon\leq 1$.
Then there exists $C=C(\Delta_{ab},m)$ such that under the $(\mathcal{D},\gamma,n)$-stochastic ball model, we have
\[
\|P_{1^\perp}M^{(a,b)}1\|_2^2
\leq \frac{4n^3\Delta_{ab}^2}{m}+Cn^3\epsilon
\]
with probability $1-e^{-\Omega_{m,\Delta_{ab},\epsilon}(n)}$.
\end{lemma}

\begin{proof}
First, a quick calculation reveals
\begin{align*}
e_i^\top M^{(a,b)}1
&=e_i^\top D^{(a,b)}1-e_i^\top D^{(a,a)}1+\frac{1}{2}\bigg(\frac{1}{n}1^\top D^{(a,a)}1-\frac{1}{n}1^\top D^{(b,b)}1\bigg),\\
\frac{1}{n}1^\top M^{(a,b)}1
&=\frac{1}{n}1^\top D^{(a,b)}1-\frac{1}{2}\bigg(\frac{1}{n}1^\top D^{(a,a)}1+\frac{1}{n}1^\top D^{(b,b)}1\bigg),
\end{align*}
from which it follows that
\begin{align*}
e_i^\top P_{1^\perp}M^{(a,b)}1
&=e_i^\top M^{(a,b)}1-\frac{1}{n}1^\top M^{(a,b)}1\\
&=\bigg(e_i^\top D^{(a,b)}1-\frac{1}{n}1^\top D^{(a,b)}1\bigg)-\bigg(e_i^\top D^{(a,a)}1-\frac{1}{n}1^\top D^{(a,a)}1\bigg)\\
&=e_i^\top P_{1^\perp}(D^{(a,b)}1-D^{(a,a)}1).
\end{align*}
As such, we have
\begin{align}
\nonumber
\|P_{1^\perp}M^{(a,b)}1\|_2^2
&=\|P_{1^\perp}(D^{(a,b)}1-D^{(a,a)}1)\|_2^2\\
\label{eq.two terms to bound}
&=\|D^{(a,b)}1-D^{(a,a)}1\|_2^2-\|P_1(D^{(a,b)}1-D^{(a,a)}1)\|_2^2.
\end{align}
To bound the first term, we apply the triangle inequality over Lemma~\ref{lemma.difference of distances}:
\begin{equation}
\label{eq.bound of first term 1}
\|D^{(a,b)}1-D^{(a,a)}1\|_2
\leq 4n\|p\|_2+\|q\|_2
\leq 4n\|p\|_2+(6+2\Delta_{ab})n^{3/2}\epsilon.
\end{equation}
We proceed by bounding $\|p\|_2$.
To this end, note that the $p_i$'s are iid random variables whose outcomes lie in a finite interval (of width determined by $\Delta_{ab}$) with probability $1$.
As such, Hoeffding's inequality gives
\[
\bigg|\frac{1}{n}\sum_{i=1}^n p_i^2-\mathbb{E}p_1^2\bigg|
\leq \epsilon
\qquad
\mbox{w.p.}
\qquad
1-e^{-\Omega_{\Delta_{ab},\epsilon}(n)}.
\]
With this, we then have
\begin{equation}
\label{eq.bound of first term 2}
\|p\|_2^2
=n\bigg(\frac{1}{n}\sum_{i=1}^np_i^2-\mathbb{E}p_1^2+\mathbb{E}p_1^2\bigg)
\leq n\mathbb{E}p_1^2+n\epsilon
\end{equation}
in the same event.
To determine $\mathbb{E}p_1^2$, first take $r_1:=e_1^\top r$.
Then since the distribution of $r$ is rotation invariant, we may write
\[
p_1
=r_{a,1}^\top(\gamma_a-O_{ab})+\|\gamma_a-O_{ab}\|_2^2
=\frac{\Delta_{ab}}{2}r_1+\frac{\Delta_{ab}^2}{4},
\]
where the second equality above is equality in distribution.
We then have
\begin{equation}
\label{eq.bound of first term 3}
\mathbb{E}p_1^2
=\mathbb{E}\bigg(\frac{\Delta_{ab}}{2}r_1+\frac{\Delta_{ab}^2}{4}\bigg)^2
=\frac{\Delta_{ab}^2}{4}\mathbb{E}r_1^2+\frac{\Delta_{ab}^4}{16}.
\end{equation}
We also note that $1\geq\mathbb{E}\|r\|_2^2=m\mathbb{E}r_1^2$ by linearity of expectation, and so
\begin{equation}
\label{eq.bound of first term 4}
\mathbb{E}r_1^2\leq\frac{1}{m}.
\end{equation}
Combining \eqref{eq.bound of first term 1}, \eqref{eq.bound of first term 2}, \eqref{eq.bound of first term 3} and \eqref{eq.bound of first term 4} then gives
\begin{equation}
\label{eq.bound of first term summary}
\|D^{(a,b)}1-D^{(a,a)}1\|_2
\leq\bigg(\frac{4n^3\Delta_{ab}^2}{m}+n^3\Delta_{ab}^4+16n^3\epsilon\bigg)^{1/2}+(6+2\Delta_{ab})n^{3/2}\epsilon.
\end{equation}
To bound the second term of \eqref{eq.two terms to bound}, first note that
\begin{equation}
\label{eq.second term 1}
\|P_1(D^{(a,b)}1-D^{(a,a)}1)\|_2
=\frac{1}{\sqrt{n}}\Big|1^\top D^{(a,b)}1-1^\top D^{(a,a)}1\Big|.
\end{equation}
Lemma~\ref{lemma.difference of average distances} then gives
\begin{equation}
\label{eq.second term 2}
\Big|1^\top D^{(a,b)}1-1^\top D^{(a,a)}1\Big|
\geq 1^\top D^{(a,b)}1-1^\top D^{(a,a)}1
\geq n^2\Delta_{ab}^2-(6+4\Delta_{ab})n^2\epsilon
\end{equation}
with probability $1-e^{-\Omega_{m,\Delta_{ab},\epsilon}(n)}$.
Using \eqref{eq.two terms to bound} to combine \eqref{eq.bound of first term summary} with \eqref{eq.second term 1} and \eqref{eq.second term 2} then gives the result.
\end{proof}

\begin{lemma}
\label{lemma.bound on rho}
There exists $C=C(\gamma)$ such that under the $(\mathcal{D},\gamma,n)$-stochastic ball model, we have
\[
\rho_{(a,b)}
\geq
n^2\big(\Delta_{ab}^2-\Delta(\Delta-2)\big)-Cn^2\epsilon
\qquad
\mbox{w.p.}
\qquad
1-e^{-\Omega_{\mathcal{D},\gamma,\epsilon}(n)}.
\]
\end{lemma}

\begin{proof}
Recall from \eqref{eq.how to construct B} that
\begin{equation}
\label{eq.decomposition of rho}
\rho_{(a,b)}
=u_{(a,b)}^\top 1
=1^\top M^{(a,b)}1-nz
=1^\top M^{(a,b)}1-n\min_{\substack{a,b\in\{1,\ldots,k\}\\a\neq b}}\min(M^{(a,b)}1).
\end{equation}
To bound the first term, we leverage Lemma~\ref{lemma.difference of average distances}:
\begin{align*}
1^\top M^{(a,b)}1
&=1^\top D^{(a,b)}1-\frac{1}{2}(1^\top D^{(a,a)}1+1^\top D^{(b,b)}1)\\
&=\frac{1}{2}\Big(1^\top D^{(a,b)}1-1^\top D^{(a,a)}1\Big)+\frac{1}{2}\Big(1^\top D^{(b,a)}1-1^\top D^{(b,b)}1\Big)\\
&\geq n^2\Delta_{ab}^2-(6+4\Delta_{ab})n^2\epsilon
\end{align*}
with probability $1-e^{-\Omega_{m,\Delta_{ab},\epsilon}(n)}$.
To bound the second term in \eqref{eq.decomposition of rho}, note from Lemma~\ref{lemma.within cluster distances} that
\begin{align*}
&\min(M^{(a,b)}1)\\
&\qquad=\min\Big(D^{(a,b)}1-D^{(a,a)}1\Big)+\frac{1}{2}\bigg(\frac{1}{n}1^\top D^{(a,a)}1-\frac{1}{n}1^\top D^{(b,b)}1\bigg)\\
&\qquad\leq\min\Big(D^{(a,b)}1-D^{(a,a)}1\Big)
+\frac{1}{2}\bigg(\bigg|\frac{1}{n}1^\top D^{(a,a)}1-2n\mathbb{E}\|r\|_2^2\bigg|+\bigg|\frac{1}{n}1^\top D^{(b,b)}1-2n\mathbb{E}\|r\|_2^2\bigg|\bigg)\\
&\qquad\leq\min\Big(D^{(a,b)}1-D^{(a,a)}1\Big)+4n\epsilon
\end{align*}
with probability $1-e^{-\Omega_{\Delta_{ab},\epsilon}(n)}$.
Next, Lemma~\ref{lemma.difference of distances} gives
\[
\min\Big(D^{(a,b)}1-D^{(a,a)}1\Big)
\leq n\Delta_{ab}^2+(6+2\Delta_{ab})n\epsilon+4n\min_{i\in\{1,\ldots,n\}}r_{a,i}^\top(\gamma_a-O_{ab}).
\]
By assumption, we know $\|r\|_2\geq1-\epsilon$ with positive probability regardless of $\epsilon>0$.
It then follows that
\[
r^\top(\gamma_a-O_{ab})
\leq-\frac{\Delta_{ab}}{2}+\epsilon
\]
with some ($\epsilon$-dependent) positive probability.
As such, we may conclude that 
\[
\min_{i\in\{1,\ldots,n\}}r_{a,i}^\top(\gamma_a-O_{ab})
\leq-\frac{\Delta_{ab}}{2}+\epsilon
\qquad
\mbox{w.p.}
\qquad
1-e^{-\Omega_{\mathcal{D},\epsilon}(n)}.
\]
Combining these estimates then gives
\[
\min(M^{(a,b)}1)
\leq n\Delta_{ab}^2-2n\Delta_{ab}+(10+2\Delta_{ab})n\epsilon
\qquad
\mbox{w.p.}
\qquad
1-e^{-\Omega_{\mathcal{D},\Delta_{ab},\epsilon}(n)}.
\]
Performing a union bound over $a$ and $b$ then gives
\[
\min_{\substack{a,b\in\{1,\ldots,k\}\\a\neq b}}\min(M^{(a,b)}1)
\leq  n\Delta^2-2n\Delta+(10+2\Delta)n\epsilon
\qquad
\mbox{w.p.}
\qquad
1-e^{-\Omega_{\mathcal{D},\gamma,\epsilon}(n)}.
\]
Combining these estimates then gives the result.
\end{proof}

\begin{lemma}
\label{lemma.bound on spectral norm of psi}
Under the $(\mathcal{D},\gamma,n)$-stochastic ball model, we have
\[
\|\Psi\|_{2\rightarrow2}
\leq\bigg(\frac{(1+\epsilon)\sigma}{\sqrt{m}}+\epsilon\bigg)\sqrt{N}
\qquad
\mbox{w.p.}
\qquad
1-e^{-\Omega_{m,k,\sigma,\epsilon}(n)},
\]
where $\sigma^2:=\mathbb{E}\|r\|_2^2$ for $r\sim\mathcal{D}$.
\end{lemma}

\begin{proof}
Let $R$ denote the matrix whose $(a,i)$th column is $r_{a,i}$.
Then
\[
\Psi
=R-\Big[(c_1-\gamma_1)1^\top~\cdots~(c_k-\gamma_k)1^\top\Big],
\]
and so the triangle inequality gives
\[
\|\Psi\|_{2\rightarrow2}
\leq\|R\|_{2\rightarrow2}+\Big\|\Big[(c_1-\gamma_1)1^\top~\cdots~(c_k-\gamma_k)1^\top\Big]\Big\|_{2\rightarrow2}
\leq\|R\|_{2\rightarrow2}+\bigg(n\sum_{a=1}^k\|c_a-\gamma_a\|_2^2\bigg)^{1/2},
\]
where the last estimate passes to the Frobenius norm.
For the first term, since $\mathcal{D}$ is rotation invariant, we may apply Theorem~5.41 in~\cite{vershynin11}:
\[
\|R\|_{2\rightarrow2}
\leq(1+\epsilon)\sigma\sqrt{\frac{N}{m}}
\qquad
\mbox{w.p.}
\qquad
1-e^{-\Omega_{m,\sigma,\epsilon}(n)}.
\]
For the second term, apply \eqref{eq.empirical center}.
The union bound then gives the result.
\end{proof}

\begin{proof}[Proof of Theorem~\ref{main_theorem}]
First, we combine Lemmas~\ref{lemma.bound on numerator}, \ref{lemma.bound on rho} and \ref{lemma.bound on spectral norm of psi}:
For every $\delta>0$, there exists an $\epsilon>0$ such that
\begin{align}
\nonumber
&2\|\Psi\|_{2\rightarrow2}^2+\sum_{a=1}^k\sum_{b=a+1}^k\frac{\|P_{1^\perp}M^{(a,b)}1\|_2\|P_{1^\perp}M^{(b,a)}1\|_2}{\rho_{(a,b)}}\\
\nonumber
&\qquad\qquad\leq 2\bigg(\frac{1+\epsilon}{\sqrt{m}}+\epsilon\bigg)^2nk+\sum_{a=1}^k\sum_{b=a+1}^k\frac{4n^3\Delta_{ab}^2/m+Cn^3\epsilon}{n^2(\Delta_{ab}^2-\Delta(\Delta-2))-Cn^2\epsilon}\\
\label{eq.main_thm_1}
&\qquad\qquad\leq n\bigg(\frac{2k}{m}+\frac{4}{m}\sum_{a=1}^k\sum_{b=a+1}^k\frac{\Delta_{ab}^2}{\Delta_{ab}^2-\Delta(\Delta-2)}+\delta\bigg)
\end{align}
with probability $1-e^{-\Omega_{\mathcal{D},\gamma,\epsilon}(n)}$.
Next, the uniform bound $\Delta_{ab}\geq\Delta$ implies
\[
\frac{\Delta_{ab}^2}{\Delta_{ab}^2-\Delta(\Delta-2)}
=\frac{1}{1-\Delta(\Delta-2)/\Delta_{ab}^2}
\leq\frac{1}{1-\Delta(\Delta-2)/\Delta^2}
=\frac{\Delta}{2}.
\]
Combining this with \eqref{eq.main_thm_1} and considering Lemma~\ref{lemma.bound on rhs}, it then suffices to have
\[
\frac{2k}{m}+\frac{4}{m}\cdot\binom{k}{2}\cdot\frac{\Delta}{2}
<\Delta(\Delta-2).
\]
Rearranging then gives
\[
\Delta
>2+\frac{2k}{m\Delta}+\frac{k(k-1)}{m},
\]
which is implied by the hypothesis since $\Delta\geq2$.
\end{proof}

\section{Proof of Theorem~\ref{theorem.spectral clustering}} \label{sec:appendix_theorem_spectral_clustering}

Put $g=\gamma/\|\gamma\|_2$ and let $z$ have unit $2$-norm.
Since $\|\Phi_0^\top z\|_2\geq\|\Phi_0^\top g\|_2$, then considering Lemma~\ref{lemma.sufficient eigenvector}, it suffices to show that the containment
\[
S_1
:=\bigg\{v\in\mathbb{S}^{m-1}:|\langle g^\top v\rangle|\leq\frac{2}{\Delta}\bigg\}
\subseteq
\bigg\{v\in\mathbb{S}^{m-1}:\|\Phi_0^\top v\|_2<\|\Phi_0^\top g\|_2\bigg\}
=:S_2
\]
holds with probability $1-e^{-\Omega_{m,\Delta}(N)}$.
To this end, we will first show that each $v\in S_1$ is also a member of $S_2$ with high probability, and then we will perform a union bound over an $\epsilon$-net of $S_1$.

We start by considering $\|\Phi^\top v\|_2$ and $\|\Phi^\top g\|_2$.
Decompose $x_i$ as either $\gamma+r_i$ or $-\gamma+r_i$ depending on whether $x_i$ belongs to the ball centered at $\gamma$ or $-\gamma$.
Let $w$ with $\|w\|_2=1$ be arbitrary.
Then
\[
(x_i^\top w)^2
=((\pm\gamma+r_i)^\top w)^2
=(\pm\gamma^\top w+r_i^\top w)^2
=(\gamma^\top w)^2\pm 2(\gamma^\top w)(r_i^\top w)+(r_i^\top w)^2,
\]
and so $\mathbb{E}(x_i^\top w)^2=(\gamma^\top w)^2+\mathbb{E}(e_1^\top r)^2$.
Linearity of expectation then gives
\[
\mathbb{E}\big[(x_i^\top g)^2-(x_i^\top v)^2\big]
=(\gamma^\top g)^2-(\gamma^\top v)^2
=\|\gamma\|^2(1-(g^\top v)^2)
\geq 1-\frac{4}{\Delta^2}.
\]
Since $|(x_i^\top g)^2-(x_i^\top v)^2|\leq 2(1+\Delta/2)^2$ almost surely, we may apply Hoeffding's inequality to get
\begin{equation}
\label{eq.hoeffding for phi}
\|\Phi^\top g\|_2^2-\|\Phi^\top v\|_2^2
=\sum_{i=1}^N \Big((x_i^\top g)^2-(x_i^\top v)^2\Big)
\geq N\bigg(1-\frac{4}{\Delta^2}\bigg)-s
\quad
\text{w.p.}
\quad
1-e^{-\Omega_{\Delta}(s^2/N)}.
\end{equation}
For a properly chosen $t$, rearranging gives that $\|\Phi^\top v\|_2<\|\Phi^\top g\|_2$.
Instead, we will use \eqref{eq.hoeffding for phi} to prove the closely related inequality $\|\Phi_0^\top v\|_2<\|\Phi_0^\top g\|_2$.
Letting $\mu$ denote the centroid of the columns of $\Phi$, we know by \eqref{eq.empirical center} that $\|\mu\|_2\leq\delta$ with probability $1-e^{-\Omega_{m,\delta}(N)}$.
In this event, every $w$ with $\|w\|_2=1$ satisfies
\begin{align}
\big|\|\Phi_0^\top w\|_2-\|\Phi^\top w\|_2\big|
\nonumber
&=\big|\|(\Phi+\mu1^\top)^\top w\|_2-\|\Phi^\top w\|_2\big|\\
\label{eq.phi0}
&=\big|\|\Phi^\top w+1\mu^\top w\|_2-\|\Phi^\top w\|_2\big|
\leq\|1\mu^\top w\|_2
\leq\sqrt{N}\delta.
\end{align}
Furthermore, 
\[
\|\Phi_0^\top w\|_2
=\|(\Phi-\mu1^\top)^\top w\|_2
\leq\|\Phi w\|_2+\|1\mu^\top w\|_2
\leq\sqrt{N}\bigg(\frac{\Delta}{2}+1+\|\mu\|_2\bigg),
\]
where the last inequality follows from Cauchy--Schwarz along with the fact that $\|x_i\|_2\leq\Delta/2+1$ for every $i$.
Taking a supremum over $w$ then gives
\begin{equation}
\label{eq.phi0 2 norm}
\|\Phi_0^\top\|_{2\rightarrow 2}
\leq\sqrt{N}\bigg(\frac{\Delta}{2}+1+\|\mu\|_2\bigg)
\leq\sqrt{N}\bigg(\frac{\Delta}{2}+1+\delta\bigg)
\quad
\text{w.p.}
\quad
1-e^{-\Omega_{m,\delta}(N)}.
\end{equation}
In \eqref{eq.hoeffding for phi}, pick $s=(N/2)(1-4/\Delta^2)=:c_1(\Delta)N$.
Then taking a union bound with \eqref{eq.phi0} gives
\[
\big(\|\Phi_0^\top v\|_2-\sqrt{N}\delta\big)^2
\leq\|\Phi^\top v\|_2^2
\leq\|\Phi^\top g\|_2^2c_1(\Delta) N
\leq\big(\|\Phi_0^\top g\|_2+\sqrt{N}\delta\big)^2-c_1(\Delta) N
\]
with probability $1-e^{-\Omega_{m,\Delta,\delta}(N)}$.
Expanding both sides and rearranging then gives
\begin{align*}
\|\Phi_0^\top v\|_2^2
&\leq\|\Phi_0^\top g\|_2^2+2\sqrt{N}\delta\big(\|\Phi_0^\top v\|_2+\|\Phi_0^\top g\|_2\big)-c_1(\Delta) N\\
&\leq\|\Phi_0^\top g\|_2^2-\underbrace{\bigg(c_1(\Delta)-4\delta\bigg(\frac{\Delta}{2}+1+\delta\bigg)\bigg)}_{c_2(\Delta)} N,
\end{align*}
where the last step follows from \eqref{eq.phi0 2 norm}.
Thus, picking $\delta=\delta(\Delta)$ sufficiently small ensures $c_2(\Delta)>0$.
Since $c_2(\Delta)N\leq\|\Phi_0^\top g\|_2^2-\|\Phi_0^\top v\|_2^2=(\|\Phi_0^\top g\|_2+\|\Phi_0^\top v\|_2)(\|\Phi_0^\top g\|_2-\|\Phi_0^\top v\|_2)$, we further have
\[
\|\Phi_0^\top g\|_2-\|\Phi_0^\top v\|_2
\geq\frac{c_2(\Delta)N}{\|\Phi_0^\top g\|_2+\|\Phi_0^\top v\|_2}
\geq c_3(\Delta)\sqrt{N},
\]
where the last inequality takes $c_3(\Delta):=c_2(\Delta)/(\Delta/2+1+\delta)$, following \eqref{eq.phi0 2 norm}.

At this point, we know that if $v\in S_1$, then $v\in S_2$ with probability $1-e^{-\Omega_{m,\Delta}(N)}$.
It remains to perform a union bound over an $\epsilon$-net of $S_1$ to conclude that $S_1\subseteq S_2$ with high probability.
To this end, pick $\epsilon<c_3(\Delta)/(\Delta/2+1+\delta)$, consider an $\epsilon$-net $\mathcal{N}_\epsilon$ of $S_1$, and suppose
\begin{equation}
\label{eq.success event}
\|\Phi_0^\top v\|_2\leq\|\Phi_0^\top g\|_2-c_3(\Delta)\sqrt{N}
\qquad
\forall v\in\mathcal{N}_\epsilon.
\end{equation}
Then for every $x\in S_1$, there exists $v\in\mathcal{N}_\epsilon$ such that $\|x-v\|_2\leq\epsilon$, and so \eqref{eq.phi0 2 norm} gives
\[
\|\Phi_0^\top x\|_2
\leq\|\Phi_0^\top\|_{2\rightarrow2}\|x-v\|_2+\|\Phi_0^\top v\|_2
\leq\sqrt{N}\bigg(\frac{\Delta}{2}+1+\delta\bigg)\epsilon+\|\Phi_0^\top g\|_2-c_3(\Delta)\sqrt{N}
<\|\Phi_0^\top g\|_2,
\]
as desired.
To measure the probability of the success event \eqref{eq.success event}, a standard volume comparison argument establishes the existence of an $\epsilon$-net of size $|\mathcal{N}_\epsilon|\leq(1+2/\epsilon)^m$; see Lemma~5.2 in~\cite{vershynin11}.
As such, the union bound gives that \eqref{eq.success event} occurs with probability $1-e^{-\Omega_{m,\Delta}(N)}$.

\end{document}